%% file: MainFileSTR.tex
\documentclass[aps,longbibliography]{revtex4-1}
\usepackage{amsmath}
\usepackage{amsfonts} \usepackage{amssymb} \usepackage{bbm}
\usepackage{graphicx}  
\usepackage{psfrag}    
\usepackage{natbib}
\usepackage{subfigure}
\usepackage[english]{babel}
\usepackage{amsthm}
\usepackage{color}
\usepackage{slashed}

\input{preamble}

\newtheorem{thm}{Theorem}[section]
\newtheorem{defi}[thm]{Definition}
\newtheorem{lem}[thm]{Lemma}
\newtheorem{prop}[thm]{Proposition}
\newtheorem{cor}[thm]{Corollary}
\newtheorem{remark}[thm]{Remark}

\begin{document}

\title{Asymptotic behavior of quantum walks with spatio-temporal coin fluctuations}
\author{A. Ahlbrecht$^{1}$}
\email{andre.ahlbrecht@itp.uni-hannover.de}
\author{C. Cedzich$^{1}$}
\email{christopher.cedzich@itp.uni-hannover.de}
\author{R. Matjeschk$^{1}$}
\email{robert.matjeschk@itp.uni-hannover.de}
\author{V.B. Scholz$^{1,2}$}
\email{scholz@phys.ethz.ch}
\author{A.H. Werner$^{1}$}
\email{albert.werner@itp.uni-hannover.de}
\author{R.F. Werner$^{1}$}
\email{reinhard.werner@itp.uni-hannover.de}
\affiliation{$^{1}$Inst. f. Theoretical Physics, Leibniz Universit\"{a}t Hannover,
Appelstr. 2, 30167 Hannover, Germany \\ $^{2}$Inst. f. Theoretical Physics, ETH Zurich, Wolfgang-Pauli-Str. 27, 8093 Zurich, Switzerland}

\begin{abstract}
Quantum walks subject to decoherence generically suffer the loss of their genuine quantum feature, a quadratically faster spreading compared to classical random walks. This intuitive statement has been verified analytically for certain models and is also supported by numerical studies of a variety of examples. In this paper we analyze the long-time behavior of a particular class of decoherent quantum walks, which, to the best of our knowledge, was only studied at the level of numerical simulations before. We consider a local coin operation which is randomly and independently chosen for each time step and each lattice site and prove that, under rather mild conditions, this leads to classical behavior: With the same scaling as needed for a classical diffusion the position distribution converges to a Gaussian, which is independent of the initial state. Our method is based on non-degenerate perturbation theory and yields an explicit expression for the covariance matrix of the asymptotic Gaussian in terms of the randomness parameters.
\end{abstract}
\maketitle

\section{Introduction}
Quantum walks describe the time evolution of a single quantum particle with internal degrees of freedom, for which both space and time are discrete. We study here the case where the underlying space is assumed to be an infinite lattice of arbitrary dimension. The dynamical rule is given by a unitary operator composed of a coin operator acting on the internal degree of freedom only, in a generally site-dependent way, and a fixed shift operator translating the particle by finitely many lattice sites depending on its internal degree of freedom. We are interested in a situation were the coin is varied randomly, as a way to model the imperfections of experimental realizations. If the distribution of coins is highly peaked around a fixed one, we would expect to see a coherent walk with a linear increase of the standard deviation with the number of time steps $t$, at least for some time. In the long run, however, the randomness will be felt, and it is this regime we will study. Depending on where we put the random dependence we can distinguish four cases, summarized in the following table, and discussed in turn below.

\begin{centering}
\begin{table}[h]
\begin{tabular}{|c || c|c |}\hline
coin        &\multicolumn{2}{c|}{spatially}\\
dependence  & fixed &random\\\hline\hline
temporally  & $\sigma\sim t$ & $\sigma\sim 1$\\
fixed       &&\\\hline
temporally  & $\sigma\sim \sqrt t$ & $\sigma\sim\sqrt t$\\
random      &&\\\hline
\end{tabular}
\caption{The type of randomness of the coin operator determines the dependence of the standard deviation $\sigma$ on the number of time steps $t$ of the quantum walk. The cases where the random coin operator is fixed in time respectively space have been studied in the literature and it was shown that they lead to localization ($\sigma\sim 1$) respectively diffusive behavior ($\sigma\sim \sqrt{t}$). Our aim is to analyze the case where both types of randomness are combined and we prove that the asymptotic behavior of such quantum walks is diffusive.}
\end{table}
\end{centering}
\vskip12pt

Coherent walks, i.e., walks \emph{without any randomness}, have been found to be useful in search-like algorithms \cite{Kempe2005,ambainis-2003-1,ambainisdist,childsQW,farhi}, precisely because they spread faster than classical walks\cite{ambainis-2003-1}, which have a similar algorithmic use. They are also the simplest case of quantum simulators, since they can be understood as the one-particle sector of so-called quantum cellular automata \cite{Schumacher}, which are quantum systems on a lattice of infinitely many interacting quantum particles. There has also been done a lot of experimental work to implement quantum walks in a variety of physical setups, starting with cold atoms \cite{Bonn} and followed by experiments with trapped ions \cite{Schaetz,Schaetz2,Blatt} and photons \cite{Silberhorn}.

Temporal fluctuations are implemented by a walk operator that is \emph{random in time}, which means that for every time step a different walk operator has to be applied, keeping, however, the spatial translation invariance in each step. It is clear that we have to take the expectation value over all possible sequences of time dependent coin operators in order to model fluctuations of the coin parameters. This expectation value turns the formerly unitary time evolution into a decoherent quantum channel. This model has been studied in a number of examples \cite{Chand2,Abal,Biham,Brun,KonnoManyCoins,Buzek} and in great detail in \cite{timerandom,joyetr,JoyeMarkov} and it was shown that such a time evolution generically leads to diffusive behavior of the quantum walk, which means that the standard deviation of the position probability distribution grows proportionally to the square root of the number of time steps.

Spatial fluctuations of experimental parameters correspond to the case where the time evolution is still unitary, and the same unitary in every step, but the coin operator is \emph{random in space}. For continuous time, i.e., Hamiltonian systems this is the well known Anderson model of disordered crystals, which exhibits localization. This means that the Hamiltonian almost surely has purely discrete spectrum, and the position distribution does not spread at all. The case of quantum walks on a one-dimensional lattice subject to spatial disorder has been studied in a number of examples both numerically \cite{kendon,obuse} and theoretically \cite{Konno2009,Konno2009a,JoyeDisordered,Disordered}. These results show that, at least in one-dimensional systems, spatial disorder implies dynamical localization, meaning that after arbitrarily many time steps the quantum walker is confined to a finite region of the lattice, up to exponentially small corrections.

In this paper we examine the case where \emph{both types of disorder} appear simultaneously. Such quantum walks have been studied numerically for example in \cite{Romanelli2} and \cite{kendon} and the simulations indicate diffusive behavior. We use the coin and shift decomposition mainly to have a precise meaning for independently identically distributed randomness. For this setting we develop a general theory of the asymptotic position distribution and find diffusive scaling. More precisely, the scaled limiting distribution is exactly gaussian and independent of the initial state. This  distinguishes the present case from only temporal randomness, where we get gaussianness only in every momentum component. Since momentum is conserved, a residual dependence on the initial state remains, and since the diffusion constant depends on momentum, the resulting mixture of Gaussians is no longer a Gaussian.

Let us briefly outline the structure of this paper. We start in section \ref{sec:model} by the mathematical formulation of the model, followed by the general examination in section \ref{sec:generaltheory}. We continue our discussion by the application to a variety of examples in section \ref{sec:examples}. In section \ref{sec:generalizations} we comment on generalizations of our results to the case of more general quantum walks and we conclude in section \ref{sec:conclude} by discussing some open problems left for future research.

\section{Models of Quantum Walks with random coin}\label{sec:model}
Abstractly, quantum walks can be defined as a discrete time evolution of a quantum particle with internal state space $\KK$ moving with strictly finite propagation speed on a lattice $\Integers^s$. Usually, one also assumes translation invariance of the time evolution which then yields a structure theorem \cite{timerandom} for the class of all translation invariant and possibly decoherent quantum walks. Hence, the underlying Hilbert space is $\HH=\WSp$ and quite commonly a single time step is realized by a composition of a local \emph{coin operator} $C$ and a conditional translation operator called (state-dependent) \emph{shift operator} $S$. Throughout this paper we will assume that the shift operator $S$ is given by a unitarily implemented quantum channel, hence, if $C$ is also unitarily implemented we can represent the quantum walk by a unitary matrix $W$ acting on $\WSp$. We denote the unitary operators corresponding to the coin respectively shift again by $C$ respectively $S$, and hence, $W=S\cdot C$. To begin with we specify $S$ and $C$ in the case where both are unitarily implemented.

We denote the elements of the standard basis of $\WSp$ by $\Ket{x\otimes i}$, where $x\in\Integers^s$ labels the positions and $i=1,\ldots,\dim \KK$ labels a basis of $\mathcal K$ such that $S$ is given by
\begin{equation}
\label{Eq:ShiftDef}
S\Ket{x\otimes i}=\Ket{x+v_i\otimes i},
\end{equation}
with some vectors $v_i\in\Integers^s$. A single time step at time $t$ is generated by the unitary operator
\[
W_t=S\cdot C_t\, ,\quad C_t=\bigoplus_{x\in\Integers^s}U_{x,t}\, ,
\]
where $U_{x,t}$ is a unitary matrix of dimension $\dim \KK$ depending on the time $t$ and the lattice site $x$. Ideally, the coin operator is translation invariant and constant in time, that is $U_{x,t}=U$. In this case we can write $C=\idty_{\Integers^s}\otimes U$ and it is well-known that this generically leads to ballistic behavior of the quantum walk, that is, the standard deviation of the position distribution grows linearly with the number of time steps $\sigma(t)\sim t$.

If the coin operator at a fixed time $t$ is translation invariant, but varies in time, we have $U_{x,t}=U_t$. In this case, the interpretation of fluctuating coin parameters corresponds to a lack of controllability of the unitary $U_t$. In other words, instead of a deterministic sequence $U_1,\ldots ,U_t$ applied sequentially to an initial state we actually have to take the expectation value over all possible sequences of time dependent coin operators. In fact, we cannot control which coin operator happens at a certain time and according to quantum mechanics and its statistical nature we have to repeat the experiment several times, each with a different sequence of coin operators, in order to extract the position distribution of the quantum walk after a fixed number of time steps. Let us assume that the coins $U_t$ are distributed independently and identically in time according to some measure $\nu$ on $\UU(\KK)$, the space of unitary operators on $\KK$. We identify the underlying probability space $\Omega$ with $\UU(\KK)$ and an element $\omega\in\Omega$ uniquely determines an operator $U_\omega\in\UU(\KK)$\footnote{In a slight abuse of notation we will not distinguish between the probability measures on $\Omega$ and $\UU(\KK)$ and just use the letter $\nu$ for both of them.}. With this notation we can describe a single time step of an observable $A$ in the Heisenberg picture by the application of a decoherent quantum channel according to
\[
\Walk (A)=S^*\cdot \mathop{\mathbbm E}\limits_\omega\left[( \idty_{\Integers^s}\otimes U^*_\omega)\cdot A \cdot (\idty_{\Integers^s}\otimes U_\omega)\right]\cdot S
\]
where the expectation value is taken with respect to the probability distribution $\nu$ of the coin operators $U_\omega\in\UU(\KK)$. This case has been studied in \cite{timerandom,joyetr,JoyeMarkov} and it was shown that such a time evolution generically leads to diffusive behavior, by which we mean that the standard deviation of the position distribution grows proportional to the square root of the number of time steps $\sigma(t)\sim \sqrt{t}$.

If on the other hand the coin operator is constant in time, at least for a large number of time steps, but inhomogeneous in space, we have $U_{x,t}=U_x$. Here, the time evolution after $t$ time steps is generated by the $t$-fold concatenation of $W=SC$, where $C=\bigoplus_xU_x=\sum_x \Ketbra{x}{x}\otimes U_x$ is the coin configuration generated by the unitary matrices $U_x$. If now the fluctuation of the experimental parameters is only spatial one gets a unitary time evolution, but if on the other hand there is a temporal fluctuation in $C$ on a large time scale, comparable to the duration of a single run of statistical data collection, one needs to take the expectation value over all possible spatial realizations of the coin operator $C$. Assuming again that the matrices $U_x$ are distributed independent and identically according to some measure $\nu$ on $\UU(\KK)$ we can write the time evolution after $t$ time steps as
\[
\Walk_t (A) =\mathop{\mathbbm E}\limits_\omega\left[\left(S^*\cdot \bigoplus_x U^*_{\omega(x)}\right)^t \cdot A\cdot \left(\bigoplus_x U_{\omega(x)}\cdot S\right)^t\right]\, ,
\]
where the expectation is taken with respect to the spatial configurations of $C$, which mathematically corresponds to the infinite product measure $\nu^\infty$ of $\nu$ defined on the probability space $\prod_{\Integers^s}\UU(\KK)$. The one-dimensional case $\HH=\ell^2(\Integers)\otimes \Complex^2$ was analyzed in \cite{Disordered,JoyeDisordered} and it was shown that such a time evolution yields Anderson localization, that is, up to exponentially small corrections, the position distribution of the quantum walker has finite support on the lattice $\Integers$ under rather general assumptions on the distribution $\nu$. Those results also apply to the case where we do not average over all possible spatial configurations, in other words, almost all possible configurations already show Anderson localization and large-scale temporal fluctuations are not required for localization.

In this paper we will analyze the combination of the former mentioned cases. Fluctuations of the coin parameters are now assumed to happen in space as well as in time on the scale of a single time step. Mathematically, this means that the $t$-fold time evolution of a single run of the experiment is given by the unitary operator $W_t\cdot \ldots \cdot W_1=SC_t\cdot\ldots \cdot SC_1$. Similarly to the other models we have to take the expectation with respect to the distribution of the coin operators, but now in space and time. A crucial assumption we impose on our model is again that the coins are independent and identically distributed both in time and space according to a measure $\nu$ on $\UU(\KK)$. Consequently, one time step of the evolution can be written as
\begin{equation}
\label{Eq:WalkPosSpace}
\Walk (A) = S^* \Coin (A)S\, ,
\end{equation}
where $S$ denotes the unitary shift operator and $\Coin$ the decoherent coin operator stemming from the expectation value with respect to all possible coin realizations. The action of the averaged coin operator $\Coin$ on a generic operator $A=\sum_{x,y}\Ketbra{x}{y}\otimes A_{xy}$ can be written as
\begin{eqnarray}
\label{Eq:CoinPosSpace}
\Coin\left( \sum_{x,y}\Ketbra{x}{y}\otimes A_{xy}\right) &=& \sum_{x,y}\Ketbra{x}{y}\otimes\left( \delta_{xy} \int \nu(d\omega) U_\omega^* A_{xy} U_\omega\right.\\
&&\left.+ (1-\delta_{xy})  \widetilde U^* A_{xy} \widetilde U_{\omega} \right)\,,\quad \widetilde U=\int\nu(d\omega) U_{\omega}\,.\nonumber
\end{eqnarray}
Now since the distribution $\nu$ is independent of the lattice site $x$ it follows that $\Walk$ itself is a translation invariant operator. This is similar to the model considered in \cite{timerandom}, the crucial difference being that there the existence of a Kraus decomposition of $\Walk$ in terms of translation invariant Kraus operators was assumed. This, however, is not the case in this model, where a Kraus decomposition is given by the Kraus operators corresponding to all possible realizations of coin operators, which is a decomposition into non-translation invariant Kraus operators. We will further develop our method used in \cite{timerandom} in order to cope also with the case of fluctuations in space and time and prove that this in fact leads to diffusive behavior.

\section{The Perturbation Method}\label{sec:generaltheory}
Our goal is to determine the scaling of the standard deviation $\sigma(t)$ of the position probability in time, in particular to distinguish between ballistic ($\sigma(t)\sim t$) and diffusive ($\sigma(t)\sim \sqrt{t}$) behavior. Using perturbation theory we compute the asymptotic limit $t\rightarrow\infty$ of the position distribution.

We start this section with a description of the general theory we are going to apply to quantum walks according to \eqref{Eq:WalkPosSpace} and \eqref{Eq:CoinPosSpace}. Since our method is based on perturbation theory of infinite dimensional operators we need to verify analyticity of the series expansions explicitly. This is done in the second part of this section. The actual analysis of our quantum walk model is postponed to the third part of this section, where we summarize our results.

\subsection{General Theory}
Similarly to the approach in \cite{timerandom} we compute the characteristic function of the position distribution of the quantum walk with initial state $\rho_0$ after $t$ time steps scaled by $\varepsilon$
\begin{equation}
\label{Eq:CharFunct}
C_{t,\varepsilon}(\lambda)=\tr(\rho_0 \Walk^t(\exp^{\ii \varepsilon\lambda \cdot Q}))
\end{equation}
and determine the limit $C(\lambda)=\lim\limits_{t\rightarrow\infty}C_{t,\varepsilon}(\lambda)$, where $\varepsilon$ is either chosen as $1/t$ in ballistic scaling or $1/\sqrt{t}$ in diffusive scaling. The method we are going to incorporate is based on perturbation theory of bounded operators \cite{Kato}. The idea is to introduce a similarity transform
\begin{alignat}{3}
	\phi:\mathcal B&(\mathcal H)\:&\longrightarrow\:&\mathcal B(\mathcal H)	\\
								&A						&\longmapsto\:&\phi(A)=Ae^{i\varepsilon\lambda Q}\nonumber
\end{alignat}
and define the operator
\begin{equation}
\label{Eq:WalkEps}
\mathbf W_\varepsilon:=\phi^{-1}\circ\mathbf W\circ\phi
\end{equation}
on $\mathcal B(\mathcal H)$, i.e. $\Walk_\varepsilon(A):=\Walk(A\exp^{\ii \varepsilon \lambda \cdot Q})\exp^{-\ii \varepsilon \lambda \cdot Q}$, and rewrite
\begin{equation}
\label{Eq:WEps}
\Walk^t(\exp^{\ii \varepsilon \lambda \cdot Q}) = \Walk_\varepsilon ^t (\idty)\exp^{\ii \varepsilon \lambda \cdot Q}\, .
\end{equation}
When inserting this into $C(\lambda)$ we can neglect the factor $\exp^{\ii \varepsilon \lambda \cdot Q}$ since $\rho_0$ is a trace-class operator and hence
\begin{align}
\label{Eq:CharFuncEps}
	C(\lambda)=\lim_{t\rightarrow\infty}\tr\left(\rho_0 \Walk_\varepsilon^t(\mathbbm1)\exp^{\ii \varepsilon\lambda \cdot Q}\right)&=\tr\left(\lim_{t\rightarrow\infty}\left(\rho_0 \Walk_\varepsilon^t(\mathbbm1)\exp^{\ii \varepsilon\lambda \cdot Q}\right)\right)	\\
			&=\tr\left( \lim_{t\rightarrow\infty}\left(\Walk_\varepsilon^t(\mathbbm1)\right)\lim_{t\rightarrow\infty}\left(\rho_0\exp^{\ii \varepsilon\lambda \cdot Q}\right)\right)	\nonumber\\
			&=\tr\left(\rho_0 \lim_{t\rightarrow\infty}\Walk_\varepsilon^t(\mathbbm1)\right) .\nonumber
\end{align}
In \eqref{Eq:CharFuncEps} we assumed that the limit of $\Walk_\varepsilon^t(\idty)$ for $t\rightarrow\infty$ exists in operator norm with appropriate scaling of $\varepsilon$. In fact, our goal is to determine this limiting operator via perturbation theory and by inserting it into \eqref{Eq:CharFuncEps} we get the characteristic function of the asymptotic distribution of $\Walk$ and $\rho_0$. We interpret $\varepsilon$ as a perturbation parameter, so we act with a high power of the perturbed operator $\Walk_\varepsilon$ on the eigenvector $\idty$ of the unperturbed operator $\Walk=\Walk_0$.

Before we deepen our analysis let us sketch the results to be expected. The operator $\Walk_\varepsilon$ approaches $\Walk$ as $\varepsilon\rightarrow 0$. Moreover, the perturbed eigenvector $A_\varepsilon$, obeying $\Walk_\varepsilon (A_\varepsilon)=\mu_\varepsilon A_\varepsilon$, approaches $\idty$ as $\varepsilon\rightarrow 0$, hence, we expect
\begin{equation}
\Walk_\varepsilon^t(\idty) \mathop{\longrightarrow}\limits_{t\rightarrow\infty}\mu_\varepsilon^t \idty \,.
\end{equation}
Consequently, the characteristic function is given by $C(\lambda)=\lim_{t\rightarrow\infty}\mu_\varepsilon^t$, which, in contrast to the case considered in \cite{timerandom}, is always independent of the initial state $\rho_0$. From the perturbation expansion $\mu_\varepsilon=1+\mu'\varepsilon + \mu'' \varepsilon^2/2 +\ldots$ with $\mu'=\ii v\cdot\lambda$ and $\mu''=-\lambda^T\cdot D\cdot \lambda$ we get in ballistic scaling $\varepsilon=1/t$ the characteristic function
\begin{equation}
C_{1/t}(\lambda)=\exp^{\ii v\cdot\lambda}\,,
\end{equation}
which is the characteristic function of a point mass at $v$ corresponding to a constant drift with velocity $v$. If $\mu'=0$ we can look at the diffusive scaling $\varepsilon=1/\sqrt{t}$ of $\Walk$ which yields
\begin{equation}
C_{1/\sqrt{t}}(\lambda)=\exp^{-\lambda^T\cdot D\cdot \lambda}\,
	\label{eq:DiffScalingCharFctn}
\end{equation}
corresponding to a position distribution which is a Gaussian with covariance matrix $D$.

\subsection{Analytic Perturbation Theory for Quantum Walks}
To begin with let us put this problem rigorously into the context of perturbation theory. It is convenient to consider vectors $\psi=\sum_x \Ketbra{x}{x}\otimes \psi_x\in\WSp$ as functions $\psi:x\mapsto \psi_x$, in other words, we identify $\WSp$ with the set $\ell^2(\Integers^s,\KK)$ of all $\KK$-valued square summable functions on $\Integers^s$. Then, a translation by a vector $y\in \Integers^s$ on $\WSp$, which we denote by $T_y$, can be defined via $T_y \psi : x\mapsto \psi_{x-y}$. With the help of these $T_y$ we define translations of bounded operators $A\in \BB(\WSp)$ by $\tau_y(A)=T_y A T_y^*$ and denote the set of all translation invariant bounded operators on $\WSp$ by $\TOps\subset\mathcal B(\WSp)$. The defining equation for $\TOps$ is
\begin{equation}
A \in \TOps \quad\Leftrightarrow \quad \tau_y (A)=A \quad \forall \,y\in\Integers^s \,.
\end{equation}

Now, let us argue why $\TOps$ constitutes a vector space on which $\Walk_\varepsilon$ acts, which is expressed by $\tau_y \circ \Walk_\varepsilon=\Walk_\varepsilon\circ \tau_y$ for all $y\in\Integers^s$. Indeed, though the similarity transform $\phi$ on $\BB(\WSp)$ does not preserve translation invariance, more precisely $\tau_y\circ\phi = \exp^{\ii \varepsilon \lambda\cdot y}\phi\circ \tau_y$, the operator $\Walk_\varepsilon=\phi^{-1}\circ \Walk \circ \phi$ commutes with translations $\tau_y$ due to the fact that $\Walk$ preserves translation invariance and the two appearing phase factors cancel each other.

The analysis of translation invariant operators $A\in\TOps$ and maps acting on $\TOps$ is much simplified by introducing the Fourier transform on $\WSp$ via
\begin{equation}
\label{Eq:Fourier}
(\Fourier \psi)(p)=(2\pi)^{-\frac{s}{2}}\sum_{x\in\Integers^s} \exp^{\ii x\cdot p} \psi_x\, ,\quad p\in [0,2\pi )^s.
\end{equation}
Then each $A\in\TOps$ becomes a multiplication operator in Fourier space, i.e. it exists a unique matrix valued function $p\mapsto A(p)\in\mathcal B(\mathcal K)$ such that $(\mathcal FA\psi)(p)=A(p)\psi(p)$ \cite{SteinWeiss}, and vice versa. This leads to the following sesqui-linear form on $\TOps$
\begin{equation}
\label{Eq:ScalarProduct}
\Scp{A}{B}:= \frac{1}{(2\pi)^s} \int\limits_{[0,2\pi)^s}\! d^s p\, \frac{1}{\dim \KK}\tr_\KK A^*(p) B(p)\, ,
\end{equation}
which in fact is even a scalar product turning $\TOps$ into a separable Hilbert space. Another way of interpreting \eqref{Eq:ScalarProduct} is to observe that a translation invariant bounded operator $A\in\TOps$ is fully characterized by its matrix entries at the origin, i.e. if we expand $A$ in position basis as $A=\sum_{x,y\in\Integers^s}\Ketbra{x}{y}\otimes A_{xy}$ with $A_{xy}\in\BB(\KK)$ then $A_{xy}=A_{0,x-y}$. The corresponding multiplication operator in Fourier space is now given by $A(p)=\sum_{x\in\Integers^s}\exp^{\ii p\cdot x}A_{x0}$ and therefore we get the alternative expression
\begin{equation}
\Scp{A}{B}=\tfrac{1}{\dim\KK}\sum_{x\in\Integers^s}\tr_\KK A_{x0}^*B_{x0}\,,
\end{equation}
which is finite since $A$ and $B$ are bounded operators. We denote the norm on $\TOps$ induced by this scalar product by $\Norm{.}$, thus, $\Norm{A}^2=\Scp{A}{A}$ for all $A\in\TOps$.

Now, after we have introduced the Hilbert space $\TOps$ we consider the specific form of quantum walks according to \eqref{Eq:ShiftDef}, \eqref{Eq:WalkPosSpace} and \eqref{Eq:CoinPosSpace}. The shift operator $S$ is represented in momentum space by conjugation with the $p$-dependent $\dim\KK$-dimensional matrix
\begin{equation}
\label{Eq:Shift}
S(p)=
\left(
\begin{array}{ccc}
\exp^{\ii v_1\cdot p} & 0 &\dots \\
0 & \ddots & \\
\vdots & & \exp^{\ii v_{\dim\KK}\cdot p}
\end{array}
\right)\, ,
\end{equation}
hence, the operator $\Walk$ acts on a translation invariant bounded operator $A(p)$ in the following way.
\begin{eqnarray}
\label{Eq:WalkMom}
\Walk(A)(p) &=& S(p)^* \left( \int\nu(d\omega) U_\omega^*A_0 U_\omega \right.\\
&&\left.+ \widetilde U^* (A(p)-A_0)\widetilde U\right)S(p) \nonumber
\end{eqnarray}
In this equation $A_0$ denotes the $p$-independent term in $A(p)$, which can be represented as $A_0=(2\pi)^{-s}\int d^s p A(p)$. The modified operator $\Walk_\varepsilon$ is now given by
\begin{eqnarray}
\Walk_\varepsilon(A)(p) &=& S(p)^* \left( \int\nu(d\omega) U_\omega^*A_0 U_\omega \right.\\
&&\left.+ \widetilde U^* (A(p)-A_0)\widetilde U\right)S(p+\varepsilon\lambda) \, .\nonumber
\end{eqnarray}
where the momentum shift of $\varepsilon\lambda$ arises from the definition of the perturbed walk operator $\Walk_\varepsilon$ and the fact that the coin operator $\Coin$ commutes with $\exp^{\ii \varepsilon \lambda\cdot Q}$, together with the equation
\begin{equation}
	\exp^{i\varepsilon\lambda Q}S(p)\exp^{-i\varepsilon\lambda Q}=S(p+\varepsilon\lambda)\, .
\end{equation}
Our goal is to apply non-degenerate perturbation theory to the operator $\Walk_\varepsilon$, in particular, the aim is to determine the first and second order equations of the perturbation theory in $\varepsilon$. The correctness of the results obtained by equating coefficients of powers of $\varepsilon$ is in fact non-trivial since $\Walk_\varepsilon$ is defined on an infinite dimensional Hilbert space $\TOps$. Hence, we first have to establish analyticity of $\Walk_\varepsilon$, the eigenvector $A_\varepsilon$ and the corresponding eigenvalue $\mu_\varepsilon$ with $A_\varepsilon\rightarrow \idty$ and $\mu_\varepsilon\rightarrow 1$ as $\varepsilon\rightarrow 0$. This can be done by using the following theorem, which is an adaption of the well-known theorem of Kato and Rellich \cite{Kato,SimonPerturbationTheory} to our setting.
\begin{thm}[Kato-Rellich]
\label{thm:KatoRellich}
Assume that the operator $\Walk_\varepsilon$ is bounded on $\TOps$ and the limit of the difference quotient
\[
\lim_{\Delta\rightarrow 0}\frac{\Walk_{\varepsilon+\Delta}-\Walk_\varepsilon}{\Delta}\, ,
\]
which we then call the derivative of $\Walk_\varepsilon$ at $\varepsilon$, exists in operator norm for all $\varepsilon$ in an open subset of $\Complex$ containing the origin. If the eigenvalue equation $\Walk (\idty)=\idty$ of the unperturbed operator $\Walk=\Walk_0$ is non-degenerate, then, for small enough $\varepsilon$, there exists an analytic eigenvector $A_\varepsilon$ with analytic and non-degenerate eigenvalue $\mu_\varepsilon$ such that $A_\varepsilon\rightarrow \idty$ and $\mu_\varepsilon\rightarrow 1$ as $\varepsilon\rightarrow 0$.
\end{thm}
For a proof of this theorem we refer to \cite{Kato,SimonPerturbationTheory}. Now, we are left to prove that $\Walk_\varepsilon$ is bounded, differentiable and the eigenvalue $1$ of $\Walk$ is non-degenerate. In order to ensure the non-degeneracy of the perturbation theory we assume that some power of $\Walk$ is strictly contractive on the orthogonal complement of $\idty$. Below we identify a class of quantum walks $\Walk$ which are strictly contractive on their own, i.e. $\Norm{\Walk(A)}<\Norm{A}$ for all $A\perp\idty$.

The following proposition in conjunction with theorem \ref{thm:KatoRellich} assures the analyticity of the perturbation theory of $\Walk_\varepsilon$.
\begin{prop}
\label{Prop:BasicProp}
Let $\Walk$ be a quantum walk according to equations \eqref{Eq:ShiftDef}, \eqref{Eq:WalkPosSpace} and \eqref{Eq:CoinPosSpace} and $\Walk_\varepsilon$ be defined by \eqref{Eq:WalkEps}. Assume that some power of $\Walk$ is strictly contractive on $\{\idty\}^\perp$, that is, there exists $n\in\Naturals$ such that $\Norm{\Walk^n(A)}<\Norm{A}$ for all $A\perp\idty$. Then we have the following conclusions:
\begin{itemize}
\item[i)] The derivative of $\Walk_\varepsilon$ exists in operator norm for all $\varepsilon\in\Complex$.
\item[ii)] For all $\varepsilon\in\Complex$ the operator $\Walk_\varepsilon$ is bounded with $\Norm{\Walk_\varepsilon}_{op}\leq\max\limits_i \vert\exp^{\ii \varepsilon\lambda\cdot v_i}\vert$.
\item[iii)] The eigenvalue $1$ of $\Walk$ is non-degenerate.
\end{itemize}
\end{prop}
\begin{proof}
It is easy to check that the operator $\Walk_\varepsilon'$ defined via its Fourier transform
\[
\Walk_\varepsilon' (A)(p) = \Walk_\varepsilon (A)(p)\cdot \ii \Lambda\, ,
\]
where $\Lambda$ is the $\dim\KK$-dimensional diagonal matrix with matrix entries $\lambda\cdot v_i$ on the diagonal, is indeed the operator norm limit of the difference quotient
\[
\lim_{\Delta\rightarrow 0}\frac{\Walk_{\varepsilon +\Delta}-\Walk_\varepsilon}{\Delta} \, ,
\]
which proves \emph{i)}.

By \eqref{Eq:Shift} we have $S(p+\varepsilon\lambda)= S(p) \cdot S(\varepsilon \lambda)$, which implies
\[
\Scp{\Walk_\varepsilon (A)}{\Walk_\varepsilon (A)} =\Scp{\Walk (A)}{\Walk(A)S(\varepsilon\lambda)S^*(\varepsilon\lambda)}\leq \, \max\limits_i\vert \exp^{\ii \varepsilon \lambda\cdot v_i}\vert^2\Scp{\Walk (A)}{\Walk(A)},
\]
and hence, by the basic definition of the operator norm, we have $\Norm{\Walk_\varepsilon}_{op}\leq\max\limits_i\vert \exp^{\ii \varepsilon \lambda\cdot v_i}\vert\Norm{\Walk}_{op}$. Now, \emph{ii)} follows from $\Norm{\Walk}_{op}=1$, which we prove next.

By writing $\Walk (A)=S^* \Coin(A) S$ and observing that $S$ is unitary we get $\Norm{\Walk (A)}=\Norm{\Coin(A)}$, and hence $\Norm{\Walk}_{op}=\Norm{\Coin}_{op}$.
We denote the first part of the coin operator in \eqref{Eq:CoinPosSpace} by $T$, that is,
\[
T(A)=\int\nu(d\omega) U_\omega^* A U_\omega
\]
and hence
\[
\Coin(A)(p)= T(A_0) + \widetilde U^* (A(p)-A_0)\widetilde U \, .
\]
The Hilbert space $\TOps$ can be decomposed into a direct sum of two orthogonal subspaces $\mathcal{T}_0$ and $\mathcal{T}_0^\perp$ defined via
\begin{eqnarray}\label{Eq:CoinSubspaces}
\mathcal{T}_0 &=& \{ A\in\TOps \,:\, A=A_0\}\\
\mathcal{T}_0^\perp &=& \{ A\in\TOps \,:\, A_0 = 0 \}\nonumber \, .
\end{eqnarray}
Clearly, $\Coin(\mathcal{T}_0)\subset \mathcal{T}_0$ and $\Coin(\mathcal{T}_0^\perp)\subset \mathcal{T}_0^\perp$, from which it also follows that $\Walk(\mathcal{T}_0) \perp \Walk(\mathcal{T}_0^\perp)$. Hence,
\[
\Norm{\Walk}_{op}=\Norm{\Coin}_{op}=\mathrm{max}\{\Norm{T}_{op},\Norm{\widetilde U^*.\widetilde U}_{op}\}\,,
\]
where we consider $T$ respectively $\widetilde U^*.\widetilde U$ as map on $\mathcal{T}_0$ respectively $\mathcal{T}_0^\perp$. First, let us note that
\begin{eqnarray*}
\Norm{T(A)}^2&\leq &\frac{1}{\dim \KK}\int\int \nu(d\omega)\nu(d\omega')\vert\tr_\KK U_\omega^* A^* U_\omega U_{\omega'}^* A U_{\omega'} \vert \\
&\leq &\frac{1}{\dim \KK}\int\int \nu(d\omega)\nu(d\omega')\sqrt{\tr _\KK A^* A }\sqrt{\tr_\KK  A^* A } \\
&=&\Norm{A}^2\, ,
\end{eqnarray*}
in other words, $\Norm{T}_{op}\leq 1$. Now, let $A$ and $B$ be positive operators on $\KK$, then we have
\[
\tr_\KK A B \leq \Norm{A}_{op} \cdot\tr_\KK B\, ,
\]
and hence, by applying this inequality to $\tr_\KK \Coin(A)^*\Coin(A)$, we get
\begin{eqnarray}
\label{Eq:NormUStarU}
\tr_\KK \widetilde U^* A^* \widetilde U\widetilde U^* A \widetilde U &=& \tr_\KK \widetilde U\widetilde U^* A^* \widetilde U\widetilde U^*A \\
&\leq & \Norm{\widetilde U\widetilde U^*}_{op}\cdot\tr_\KK A^*\widetilde U\widetilde U^*A \nonumber\\
&\leq & \Norm{\widetilde U \widetilde U^*}_{op}^2 \cdot \tr_\KK AA^* \nonumber\\
&\leq & \Norm{\widetilde U}_{op}^4 \cdot \tr_\KK AA^* \, .\nonumber
\end{eqnarray}
Clearly, $\Norm{\widetilde U}_{op}\leq \int\nu(d\omega)\Norm{U_\omega}_{op}=1$ which finally proves $\Norm{\widetilde U^*.\widetilde U}_{op}\leq 1$, and hence, $\Norm{\Coin}_{op}\leq 1$.

Any eigenvector of $\Walk$ is also an eigenvector of $\Walk^n$ with eigenvalue raised to the $n$-th power, hence, the contractivity of $\Walk^n$ yields statement \emph{iii)}.
\end{proof}
Of course, the contractivity of $\Walk$ or a power of it depends on the probability distribution $\nu$ of the coin operators $U_\omega$. In particular, the operator $\Walk$ itself is strictly contractive if the operators $U_{\omega'}^*U_\omega$ with $\omega',\omega\in\Omega$ fulfill the following definition.
\begin{defi}
\label{Def:Irreducible}
A set of matrices $\{M_i\,:\, i\in I\}$, where $I$ is an index set and each $M_i$ acts on $\KK$, is said to be irreducible if any invariant subspace is trivial, that is, if $\mathcal{S}$ is a subspace of $\KK$ such that $M_i \mathcal{S}\subset \mathcal{S}$ for all $i\in I$, then we must have $\mathcal{S}=\{0\}$ or $\mathcal{S}=\KK$.
\end{defi}

\begin{prop}
\label{prop:IrredContr}
Let $\Walk$ be a quantum walk according to equations \eqref{Eq:ShiftDef}, \eqref{Eq:WalkPosSpace} and \eqref{Eq:CoinPosSpace}. If the coin operators $U_\omega$ on which the measure $\nu$ is supported are such that the set $\{U_{\omega'}^*U_\omega\,:\,\omega,\omega'\in\Omega\}$ is irreducible on $\KK$, then $\Walk$ is strictly contractive on $\{\idty\}^\perp$, that is, $\Norm{\Walk(A)}<\Norm{A}$ for all $A\perp \idty$.
\end{prop}
\begin{proof}
Suppose the set $\{U_{\omega'}^*U_\omega\,:\,\omega,\omega'\in\Omega\}$ is irreducible. It follows from the singular value decomposition of $\widetilde U$ that there exist normalized vectors $\phi,\, \psi\in\KK$ such that
\[
\Norm{\widetilde U}_{op}= \Scp{\phi}{\widetilde U\psi} = \int\nu(d\omega)\Scp{\phi}{ U_\omega \psi}\, ,
\]
and hence, if $\Norm{\widetilde U}_{op}=1$ we must have $\Scp{\phi}{U_\omega \psi}=1$ for all $\omega$. This implies $\Scp{\phi}{U_{\omega'}^*U_\omega \phi}=1$, hence, $\phi$ is an eigenvector of $U_{\omega'}^*U_\omega$ for arbitrary $\omega$ and $\omega'$, which is forbidden since the set $\{U_{\omega'}^*U_\omega\,:\,\omega,\omega'\in\Omega\}$ is assumed to be irreducible. Consequently, $\Norm{\widetilde U}_{op}<1$ and by \eqref{Eq:NormUStarU} $\Norm{\widetilde U^*.\widetilde U}_{op}<1$.

Now, let $A$ be $p$-independent, i.e. $A_0=A$, such that $A\perp \idty$. Assume $\Norm{T(A)}=\Norm{A}$, where $T$ denotes the diagonal part of $\Coin$, that is, $T(A)=\int\nu(d\omega)U_\omega^*AU_\omega$. This implies
\[
\int\int\nu(d\omega)\nu(d\omega')\tr_\KK U_{\omega'}^*A^*U_{\omega'}  U_\omega^*A U_\omega = \tr_\KK A^*A\, .
\]
On the other hand, we can estimate
\begin{eqnarray*}
\int\int\nu(d\omega)\nu(d\omega')\tr_\KK U_{\omega'}^*A^*U_{\omega'}  U_\omega^*A U_\omega &\leq &\int\int\nu(d\omega)\nu(d\omega')\vert\tr_\KK U_{\omega'}^*A^*U_{\omega'}  U_\omega^*A U_\omega\vert\\
&= &\int\int\nu(d\omega)\nu(d\omega')\vert\tr_\KK A^*(U_\omega U_{\omega'}^*)^*A U_\omega U_{\omega'}^*\vert\\
&\leq &\int\int\nu(d\omega)\nu(d\omega') \sqrt{\tr_\KK A^*A}\sqrt{ \tr_\KK A^*A}=\tr_\KK A^*A\, .
\end{eqnarray*}
It follows from $\Norm{T(A)}=\Norm{A}$ that this inequality is actually an equality, hence, we must have
\[
\vert\tr_\KK A^*(U_\omega U_{\omega'}^*)^*A U_\omega U_{\omega'}^*\vert=\tr_\KK A^*A\,,
\]
which, by the Cauchy-Schwarz inequality, means that $A$ and $(U_\omega U_{\omega'}^*)^*A U_\omega U_{\omega'}^*$ must be proportional for arbitrary $\omega$ and $\omega'$. Thus, $(U_\omega U_{\omega'}^*)^*A U_\omega U_{\omega'}^*=c_{\omega,\omega'}\cdot A$ with $\vert c_{\omega,\omega'}\vert =1$ and
\[
\Norm{T(A)}^2=\frac{1}{\dim \KK}\int\int\nu(d\omega)\nu(d\omega')\tr_\KK U_{\omega'}^*A^*U_{\omega'}  U_\omega^*A U_\omega = \frac{1}{\dim \KK}\tr_\KK A^*A \int\int\nu(d\omega)\nu(d\omega')c_{\omega,\omega'}\, .
\]
The assumption $\Norm{T(A)}=\Norm{A}$ entails $c_{\omega,\omega'}=1$ for all $\omega,\omega'$. Thus, $A$ commutes with $U_\omega U_{\omega'}^*$ and since those are irreducible it follows that $A=a\cdot \idty$. This contradicts the assumption $A\perp \idty$, hence, $\Norm{T(A)}<\Norm{A}$.
\end{proof}

\subsection{Asymptotic Position Distribution via First and Second Order Perturbation Theory}\label{Sec:AsymptPosDistr}
Before we start to analyze the first and second order perturbation theory of $\Walk_\varepsilon$ to determine the asymptotic position distribution of quantum walks we have to explain why
\[
\Walk_\varepsilon^t (\idty)\rightarrow \mu_\varepsilon^t \idty
\]
as $t\rightarrow \infty$ and $\varepsilon\rightarrow 0$. First of all, we can write $\idty=A_\varepsilon +(\idty -A_\varepsilon)$ and obtain $\Walk_\varepsilon^t(\idty)=\Walk_\varepsilon^t(A_\varepsilon)+\Walk_\varepsilon^t(\idty-A_\varepsilon)$. With the operator norm estimate $\Norm{\Walk_\varepsilon}_{op} \leq 1$, which is valid for $\varepsilon\in\Reals$, we get
\[
\Norm{\Walk_\varepsilon^t (\idty) - \mu_\varepsilon^t A_\varepsilon }\leq \Norm{\idty -A_\varepsilon}\,
\]
and consequently $\Walk_\varepsilon^t (\idty) \rightarrow \mu_\varepsilon^t A_\varepsilon $ with $\varepsilon\rightarrow 0$ and the assertion follows from $A_\varepsilon \rightarrow \idty$ as $\varepsilon\rightarrow 0$.

The ballistic respectively diffusive scaling of the position distribution in the asymptotic limit is determined by the first respectively second order of the perturbed eigenvalue $\mu_\varepsilon=1+\varepsilon \mu' +\varepsilon^2/2 \mu'' +\ldots$, see also \cite{timerandom}. And indeed, in ballistic scaling, i.e. $\varepsilon =1/t$ we get the asymptotic limit of the characteristic function
\begin{equation}
\label{Eq:CharFuncBall}
C(\lambda )=\lim_{t\rightarrow \infty}\mu_{1/t}^t = \lim_{t\rightarrow \infty}\left( 1+\frac{\mu'}{t}+\ldots \right)^t = \exp^{\mu'}\, .
\end{equation}
If, however, the first order is zero $\mu'=0$ we may consider the diffusive scaling $\varepsilon=1/\sqrt{t}$ of the position distribution and get
\begin{equation}
\label{Eq:CharFuncDiff}
C(\lambda )=\lim_{t\rightarrow \infty}\mu_{1/\sqrt{t}}^t = \lim_{t\rightarrow \infty}\left( 1+\frac{\mu''}{2 t}+\ldots \right)^t = \exp^{\frac{\mu''}{2}}\, .
\end{equation}
\begin{remark}
In contrast to the case considered in \cite{timerandom}, where for some models a dependence on the initial state $\rho_0$ was observed, we see that for the present model the initial state $\rho_0$ is irrelevant.
\end{remark}
Now we have all necessary tools to compute the asymptotic position distribution. The first order correction $\mu'$ can be determined from the first order relation obtained from equating coefficients in $\Walk_\varepsilon (A_\varepsilon)=\mu_\varepsilon A_\varepsilon$, which reads
\begin{equation}
\label{Eq:FirstOrder}
\Walk'(\idty)+\Walk (A')=\mu'\idty+A' \,,
\end{equation}
where $\Walk'$ denotes the derivative of $\Walk_\varepsilon$ at $\varepsilon=0$. The solution to the eigenvector problem $\Walk_\varepsilon (A_\varepsilon)=\mu_\varepsilon A_\varepsilon$ is in general not unique. A common choice for $A_\varepsilon$ is to fix the scalar product of the perturbed eigenvector with the unperturbed eigenvector in the following way
\[
\Scp{\idty}{A_\varepsilon} =1 \quad \Rightarrow \quad \Scp{\idty}{A^{(n)}}=0\quad ,\forall\,n\in\Naturals\, ,
\]
where $A^{(n)}$ denotes the $n$-th order correction of the eigenvalue $A_\varepsilon$. That is, $A_\varepsilon$ can be expressed as
\[
A_\varepsilon = \sum_{n=0}^\infty A^{(n)}\frac{\varepsilon^n}{n!}\,.
\]
Fixing the scalar product in this way does not harm analyticity of the $A_\varepsilon$, at least for small enough $\varepsilon$, since the scalar product of an unnormalized vector with $\idty$ is an analytic function in $\varepsilon$ which is non-zero for small $\varepsilon$.

The standard approach to determine the higher order corrections to the unperturbed eigenvalue $1$ is to expand the equation $\Walk_\varepsilon (A_\varepsilon)=\mu_\varepsilon A_\varepsilon$ into a power series in $\varepsilon $ and then take the scalar product with the unperturbed eigenvector $\idty$. The choice $\Scp{\idty}{A_\varepsilon}=1$, which is equivalent to $\Scp{\idty}{A^{(n)}}=0$, implies $\Scp{\idty}{\Walk(A^{(n)})}=\Scp{\idty}{\Coin(A^{(n)})}=0$ for all $n>0$, thus, some terms in the power series expansion of the eigenvector equation already vanish when taking the scalar product with $\idty$.

By \eqref{Eq:Shift} and $\Walk_\varepsilon(A)(p)=S(p)^*\Coin(A)(p)S(p+\varepsilon \lambda)$ we can express the derivative of $\Walk_\varepsilon$ at $\varepsilon=0$ as
\[
\Walk'(A)(p)=\Walk(A)(p)\cdot \ii \Lambda \, , \quad
\Lambda=
\left(
\begin{array}{ccc}
v_1\cdot \lambda & 0 &\dots \\
0 & \ddots & \\
\vdots & & v_{\dim\KK}\cdot \lambda
\end{array}
\right)
\,,
\]
which leads to the following expression for the first order correction to the eigenvalue
\begin{eqnarray}
\label{Eq:FirstOrderFormula}
\mu' &=& \Scp{\idty}{\Walk'(\idty)}=\frac{1}{(2\pi)^s}\int d^sp\, \frac{1}{\dim\KK}\tr_\KK \Walk'(\idty) \\
&=& \frac{1}{\dim\KK}\sum_i \ii \lambda\cdot v_i =\ii \lambda\cdot \tilde v \, , \nonumber
\end{eqnarray}
where $\tilde v$ is the average of all shift vectors $v_i$. The vector $\tilde v$ is closely related to the index $\ind\, S$ of the unitary shift operator $S$, see \cite{Index} for a detailed discussion of this quantity. By definition we have
\begin{equation}
\label{Eq:Index}
\ind\, S= \sum_i v_i
\end{equation}
and consequently $\tilde v=(\dim \KK)^{-1}\ind\, S$. Clearly, the characteristic function in ballistic scaling is given by $C(\lambda)=\exp^{\ii \lambda \cdot \tilde v}$, which is the characteristic function of a point mass at position $\tilde v$. This represents a constant drift of the particle in the direction of $\ind\, S$. It should be noted, that the constant drift is not a ballistic behavior due to the quantumness of the walk. Rather, $\tilde v $ represents the part of the shift operator, which is not conditioned on the coin state. That is, in an extreme case one may consider a one-dimensional quantum walk with two-dimensional internal state space $\KK=\Complex^2$, where $v_1=v_2$, i.e. where the shift does not depend on the coin states at all. A ballistic spreading due to the quantumness of the walk does therefore not exist in this model. See remark \ref{rem:SubtractBall} below for how to investigate the diffusive order when $\tilde v \neq 0$.
\begin{remark}
\label{rem:SubtractBall}
If $\ind\,S\neq 0$ we may subtract the constant drift according to $\mu'=\ii \lambda\cdot \tilde v$ from the position variable and consider the asymptotic distribution of $\widetilde Q=Q-\tilde v\cdot t$. The characteristic function is given by
\[
C_{t,\varepsilon}(\lambda)=\tr \rho_0 \Walk^t (\exp^{\ii \varepsilon \lambda\cdot \widetilde Q })=\tr \rho_0 \exp^{-\ii \varepsilon t\lambda\cdot \tilde v}\Walk^t (\exp^{\ii \varepsilon \lambda\cdot  Q })\, ,
\]
and hence we have to consider the modified operator $\widetilde\Walk_\varepsilon(A)=\exp^{-\ii \varepsilon \lambda\cdot \tilde v}\Walk(A\exp^{\ii \varepsilon \lambda\cdot Q })\exp^{-\ii \varepsilon \lambda\cdot Q }=\exp^{-\ii \varepsilon \lambda\cdot \tilde v}\Walk_\varepsilon (A)$. The eigenvalue $\tilde \mu_\varepsilon$ is now just given by $\tilde\mu_\varepsilon = \exp^{-\ii \varepsilon \lambda\cdot \tilde v} \mu_\varepsilon$, thus
\begin{eqnarray*}
\tilde\mu_\varepsilon &=&\left(1 -\mu' \varepsilon +\frac{\mu'^2}{2}\varepsilon^2 +\ldots\right)\left(1+\mu'\varepsilon +\frac{\mu''}{2}\varepsilon^2 + \ldots\right)\\
&=& 1 + \frac{\mu''-\mu'^2}{2}\varepsilon^2 +\ldots \, ,
\end{eqnarray*}
which shows that $\tilde \mu'=0$.
\end{remark}
In the following we assume $\mu'=0$ and look at the diffusive scaling of the position distribution. To this end we need to determine the second order of the equation $\Walk_\varepsilon (A_\varepsilon)=\mu_\varepsilon A_\varepsilon$, which reads
\begin{equation}
\label{Eq:SecondOrder}
\Walk''(\idty) + 2 \Walk'(A') + \Walk(A'') = \mu''\idty + 2\mu'A' + A'' \,.
\end{equation}
The second order derivative of the walk operator amounts to
$\Walk''(A)(p)=\Walk'(A)(p)\cdot \ii \Lambda=-\Walk(A)(p)\cdot \Lambda^2$. Again, by taking the scalar product with the unperturbed eigenvector we get
\begin{eqnarray}
\label{Eq:SecondOrderFormula}
\mu'' &=& \Scp{\idty}{\Walk''(\idty)} + 2 \Scp{\idty}{\Walk'(A')}\\
&=& \frac{1}{(2\pi)^s}\int d^sp\, \frac{1}{\dim\KK}\tr_\KK \left(\Walk''(\idty)  +2 \ii \Walk(A')(p)\cdot \Lambda\right)\nonumber \\
&=& \frac{1}{\dim\KK}\sum_i -(\lambda\cdot v_i)^2 +
\frac{1}{(2\pi)^s}\int d^sp\, \frac{2\ii}{\dim\KK} \tr_\KK \left(-\Walk' (\idty) + A'(p)\right)\cdot \Lambda\nonumber \\
&=& \frac{1}{\dim\KK}\sum_i (\lambda\cdot v_i)^2 + \frac{1}{(2\pi)^s}\int d^sp\, \frac{2\ii}{\dim\KK} \tr_\KK  A'(p)\cdot \Lambda \, ,\nonumber
\end{eqnarray}
where we have used \eqref{Eq:FirstOrder} at the fourth step and $\Walk'(A)(p)=\Walk(A)(p)\cdot \ii\Lambda$ in every step. Next, we use \eqref{Eq:FirstOrder} to eliminate $A'$ in this equation. Using $\mu'=0$ and $\Walk'(\idty)(p)=\ii \Lambda$ we get
\[
\Walk(A')(p)-A'(p) = -\ii \Lambda
\]
and from $\mu'=0$, which is equivalent to $\tr_\KK \Lambda =0$, it follows that $\Scp{\idty}{\Lambda}=0$. Since the eigenvalue $1$ of $\Walk$ is non-degenerate we can invert $\Walk-\mathrm{id}$ on $\{\idty\}^\perp$ and apply it to $\Lambda \in \{\idty\}^\perp $. Denoting the pseudo-inverse by $(\Walk-\mathrm{id})^{-1}$ we get the following expression for $A'$.
\begin{equation}
\label{Eq:APrime}
A'(p)=-\ii (\Walk -\mathrm{id})^{-1}(\Lambda)(p)\,.
\end{equation}
Hence, the second order correction $\mu''$ can be written as
\begin{equation}
\label{Eq:SecondOrderFinal}
\mu''=\frac{1}{\dim\KK}\sum_i (\lambda\cdot v_i)^2 + \frac{2}{\dim\KK} \frac{1}{(2\pi)^s}\int d^sp\,\tr_\KK (\Walk -\mathrm{id})^{-1}(\Lambda)(p) \cdot \Lambda\,.
\end{equation}
Note that the dependence of $\Lambda$ on the variable $\lambda$ makes $\mu''$ a quadratic form in $\lambda$ which defines a symmetric matrix $D$ via $\mu'' = -\lambda^T\cdot D\cdot \lambda$. It is easy to see that the position distribution corresponding to the characteristic function $C(\lambda)=\exp^{\frac{\mu''}{2}}$ is a Gaussian with covariance matrix $D$.

The following theorem summarizes the results of this section.
\begin{thm}
\label{thm:MainThm}
Let $\Walk$ be a quantum walk as defined in \eqref{Eq:ShiftDef}, \eqref{Eq:WalkPosSpace} and \eqref{Eq:CoinPosSpace}. Suppose for some $n\in\Naturals$ the power $\Walk^n$ is strictly contractive on $\{\idty\}^\perp$. Then:
\begin{itemize}
\item[i)] \emph {ballistic scaling:}\/  the random variable $Q/t$, converges to a point measure at $\tilde v=\frac{1}{\dim\KK}\,\ind\,S$.
\item[ii)] \emph {diffusive scaling:}\/ assuming $\tilde v=0$ the random variable $Q/\sqrt{t}$, converges to a Gaussian with covariance matrix $D$. The defining equation for $D$ is $\mu''=-\lambda^T\cdot D\cdot \lambda$ with $\mu''$ from  \eqref{Eq:SecondOrderFinal}. Explicitly, the matrix elements of $D$ are given by the formula
\begin{equation}
\label{Eq:CovarianceMatrix}
D_{\alpha\beta} = \frac{1}{\dim\KK}\sum_i v_{i,\alpha} v_{i,\beta} +  \frac{2}{(2\pi)^s}\int d^sp\,\frac{1}{\dim\KK}\tr_\KK R_{\alpha \beta}(p)\,,\quad \alpha,\beta =1,\ldots ,s\, ,
\end{equation}
where $v_{i,\alpha}$ denotes the $\alpha$ component of the vector $v_i$ and $R_{\alpha\beta}(p)$ is defined via
\begin{equation}
\label{Eq:Ralphabeta}
R_{\alpha\beta}(p)= \frac{1}{2}\Bigl((\Walk -\mathrm{id})^{-1}(\Lambda_\alpha)(p) \cdot \Lambda_\beta + (\Walk -\mathrm{id})^{-1}(\Lambda_\beta)(p) \cdot \Lambda_\alpha\Bigr)\, .
\end{equation}
The diagonal matrices $\Lambda_\alpha$ are given by $(\Lambda_\alpha)_{ij}=\delta_{ij}\,v_{i,\alpha}$.
\end{itemize}
\end{thm}
Of course, for the asymptotic position distribution in diffusive scaling to be well-defined it is necessary that the covariance matrix $D$ is positive. This is indeed the case, which we prove now.
\begin{prop}
Given the assumptions of theorem \ref{thm:MainThm} the covariance matrix $D$ is positive, i.e., $-\mu''\geq 0$ for all $\lambda$.
\end{prop}
\begin{proof}
The proof is similar to, though simpler than, the proof of the positivity of the covariance matrix for the models considered in \cite{timerandom}. According to \eqref{Eq:SecondOrderFormula} we have
\begin{eqnarray*}
-\mu'' &=& -\Scp{\idty}{\Walk''(\idty)} - 2 \Scp{\idty} {\Walk'(A')}\\
&=& \Scp{\Lambda}{\Lambda} + 2 \Scp{\ii\Lambda}{\Walk (A')}\\
&=& \Scp{A' -\Walk(A')}{A' -\Walk(A')} + 2\Scp{A' -\Walk(A')}{\Walk (A')}\\
&=& \Scp{A'}{A'}-\Scp{\Walk(A')}{\Walk (A')}\, ,
\end{eqnarray*}
where we have used $\Walk'(A)(p)=\Walk(A)(p)\cdot \ii\Lambda$, $\Walk''(A)(p)=-\Walk(A)(p)\cdot \Lambda^2$, $\ii\Lambda =A'(p) -\Walk(A')(p)$ and the fact that $A'$ and $\Walk(A')$ are skew-hermitian, which follows from \eqref{Eq:APrime} and the fact that $\Walk(A)^*=\Walk(A^*)$ for arbitrary $A$. According to Proposition \ref{Prop:BasicProp} we have $\Norm{\Walk}_{op}\leq 1$, which proves $-\mu''\geq 0$.
\end{proof}
For practical applications it is usually sufficient to give a good approximation to the covariance matrix $D$. Such an approximation can be obtained by a power series expansion of the pseudo-inverse $(\Walk-\mathrm{id})^{-1}$. Indeed, if there exist $n\in\Naturals$ such that $\Norm{\Walk^n}_{op}<1$ we have the following convergent power series expansion for the pseudo inverse.
\begin{equation}
(\Walk-\mathrm{id} )^{-1} = -\sum_{k\in\Naturals_0}\Walk^k\quad \text{on the subspace}\quad \{\idty\}^\perp
\end{equation}
This gives us the following corollary, which can be used to approximate the numbers $R_{\alpha\beta}$ and the covariance matrix $D$.
\begin{cor}
\label{cor:CovMatPowerSeries}
Given the assumptions of theorem \ref{thm:MainThm} and $\tilde v=0$ we have the convergent power series expression
\begin{equation}
R_{\alpha\beta}(p) = -\frac{1}{2}\sum_{k\in\Naturals_0 }\left(\Walk^k(\Lambda_\alpha)(p)\cdot \Lambda_\beta +\Walk^k(\Lambda_\beta)(p)\cdot \Lambda_\alpha \right)\,.
\end{equation}
\end{cor}
\begin{proof}
Let $n\in\Naturals$ be the smallest natural number such that $\Walk^n$ is strictly contractive on $\{\idty\}^\perp$. The convergence of the Neumann series for $(\Walk-\mathrm{id})^{-1}$ can be seen from
\[
-\sum_{k\in\Naturals_0}\Walk^k = -\sum_{l=0}^{n-1} \sum_{r=0}^\infty  \Walk^{l+r\cdot n}\, .
\]
The assertion follows from $\Scp{\idty}{\Lambda}=\tilde v=0$.
\end{proof}

\section{Examples}\label{sec:examples}

In this section we apply our theory to five models of randomness and calculate their ballistic and diffusive scaling. Continuous as well as discrete distributions on the coin operations are considered and we also include an example in two dimensions.
In all cases the irreducibility condition of proposition \ref{prop:IrredContr} is violated, so we have to look at higher powers of the walk operator $\Walk$ to show simplicity of the eigenvalue $1$.

\subsection{Coins with Zero Mean $\widetilde U$}
In this subsection we assume that $\nu$ is such that $\widetilde U=\int\nu(d\omega) U_\omega =0$. Thus $\Coin(A)=\Coin(A_0)=\int\nu(d\omega)U_\omega^* A_0 U_\omega$, where we again have abbreviated $(2\pi)^{-s}\int dp \,A(p) = A_0$. If the operator $\Walk^n$ is strictly contractive for some $n\in\Naturals$, theorem \ref{thm:MainThm} is applicable and we can look at the second order correction $\mu''$ in order to determine the diffusive scaling of the position distribution. Without loss of generality we assume $\mu'=0$, see remark \ref{rem:SubtractBall}, such that the first order of the eigenvalue problem reads $\Walk(A')(p)-A'(p)= -\ii \Lambda$. Using $\Walk (A)(p) = S^*(p)\Coin (A_0) S(p)$, we get
\[
S^*(p)\Coin (A'_0) S(p)  - A'(p) =-\ii \Lambda
\]
and since $\Lambda$ commutes with $S(p)$, this is equivalent to
\[
\Coin (A'_0) +\ii \Lambda =S(p)A'(p) S^*(p) \, .
\]
Now, since the left-hand-side of this equation is independent of $p$, so is the right-hand-side. Hence, the matrix elements of $A'(p)$ are given by $\ScpWOp{i}{A'(p)}{j}=\widetilde{a}_{ij} \exp^{\ii(v_j-v_i)\cdot p}$ with $\widetilde{a}_{ij}\in\Complex$. Let $\widetilde A'$ denote the matrix with entries $\widetilde a_{ij}$, in other words, $A'(p)=S^*(p)\widetilde A' S(p)$. Then, $\widetilde A'$ can be determined from the $p$-independent equation
\[
\Coin(P(\widetilde A'))-\widetilde A'=-\ii \Lambda\, ,
\]
where $P(\widetilde A')$ is defined as
\[
\ScpWOp{i}{P(\widetilde A')}{j}=\left\{
\begin{array}{rcl}
\widetilde a_{ij} &,& \text{if}\, v_i = v_j\\
0 &,& \text{if}\, v_i\neq v_j
\end{array}
\right.\,.
\]

\subsection{Hadamard Walk with Random Reflections at Lattice Sites}
A simple example to demonstrate our techniques is the Hadamard walk on $\ell^2(\mathbbm Z)\otimes \Complex^2$ which is distorted by random reflections. Similar models have previously been studied numerically in \cite{Romanelli2,PerezRomanelli}. The intuitive picture is that at each time step some links between the lattice points are broken such that the walker cannot pass them but is reflected.
These reflections can be thought of as flips of the internal degree of freedom followed by the usual shift. Hence, the reflections can be implemented by the Pauli operator $\sigma_x$ which is applied with probability $1-w$ for $w\in(0,1)$ at each lattice site.

The basic Hadamard walk is given by the unitary $W_H=S\cdot C$, where the coin $C=\idty\otimes H$ and shift operator $S$ are defined with respect to a basis $\Ket{\pm}$ of $\Complex^2$ by
\begin{equation}
\label{Eq:Hadamard}
H=
\frac{1}{\sqrt{2}}
\left(
\begin{array}{cc}
1 & 1\\
1 & -1
\end{array}
\right)\quad,\quad
S\Ket{x\otimes\pm}=\Ket{x\pm 1\otimes\pm}\, .
\end{equation}
The Pauli operator $\sigma_x$ acts like a bit flip on $\Ket{\pm}$, i.e. $\sigma_x\Ket{\pm}=\Ket{\mp}$. Following the above, the measure $\nu_w$ on $\mathcal U(2)$ is a convex combination of the point measures $\Delta_H$ respectively $\Delta_{\sigma_x}$ supported on $H$ respectively $\sigma_x$.
\begin{equation}
	\nu_w=w\:\Delta_H+(1-w)\:\Delta_{\sigma_x}
\end{equation}
According to equation \eqref{Eq:WalkPosSpace} the walk operator is given by $\Walk (A) = S^* \Coin (A)S$ with shift operator $S$ as in \eqref{Eq:Hadamard} and coin operator
\begin{equation}
\label{Eq:RomanelliCoin}
	\Coin(A)=\sum_{x,y}\Ketbra{x}{y}\otimes\left(\delta_{xy}\left(w\, H A_{xy}H + (1-w)\,\sigma_x A_{xy}\sigma_x\right) +\left(1-\delta_{xy}\right)\widetilde UA_{xy}\widetilde U\right)
\end{equation}
where $\widetilde U=w\,H+(1-w)\,\sigma_x$. To apply our theory to this particular model, the model has to satisfy the conditions of theorem \ref{Prop:BasicProp}. Since the set $\{U_{\omega'}^*U_\omega\}=\{H\sigma_x,\:\sigma_xH\}$ is not irreducible we have to look at higher powers of $\Walk$ in order to ensure the non-degeneracy of the eigenvalue $1$.
\begin{lem}
\label{lem:BrokenLinks}
Let $w\in(0,1)$ and $\Walk=S^* \Coin S$, with $\Coin$ according to \eqref{Eq:RomanelliCoin} and $S$ according to \eqref{Eq:Hadamard}. Then $\Walk^2$ is strictly contractive on $\{\idty\}^\perp$.
\end{lem}
\begin{proof}
We decompose the coin operator again into diagonal part $T$ and conjugation with $\widetilde U$, i.e.,
\[
\Coin (A) = T(A_0) + \widetilde U^*(A-A_0)\widetilde U\, .
\]
The operator $\widetilde U=w\,H+(1-w)\,\sigma_x$ is a convex combination of two unitary irreducible operators. Since $\widetilde U$ is also hermitian we get an estimate for its operator norm by considering its eigenvalues. It is easy to see that $\Norm{\widetilde U}_{op}=1$ implies the existence of a common eigenvector of $H$ and $\sigma_x$, which is a contradiction to their irreducibility. Hence, $\Norm{\widetilde U}_{op}<1$ for all $w\in(0,1)$.

The diagonal part $T$ satisfies the eigenvalue equation $T(\sigma_y)=-\sigma_y$, hence, $\Norm{T}_{op}=1$. Similar to the proof of proposition \ref{Prop:BasicProp} we define the orthogonal subspaces $\mathcal{T}_0$ and $\mathcal{T}_0$ according to \eqref{Eq:CoinSubspaces}. Note that $\Coin (\mathcal{T}_0)\subset \mathcal{T}_0$ and $\Coin (\mathcal{T}_0^\perp)\subset \mathcal{T}_0^\perp$. It is easy to verify the relations $\Norm{\Walk^2(A)}=\Norm{\Coin(\Walk(A))}$ and $\Norm{\Coin(A)}^2=\Norm{\Coin(A_0)}^2+\Norm{\Coin(A-A_0)}^2$, hence, the existence of $A\perp\idty$ with $\Norm{\Walk^2 (A)}=\Norm{A}$ implies $A,\Walk(A)\in\mathcal{T}_0$. Writing $A$ as a linear combination of the Pauli operators $A=a_x\sigma_x+a_y\sigma_y+a_z\sigma_z$ we get the relation
\begin{equation}
\label{Eq:PauliContractive}
\Walk( A ) = w\, (a_x \sigma_z  - a_y S^*\sigma_yS + a_z S^*\sigma_x S)+(1-w)(a_x S^*\sigma_x S - a_yS^*\sigma_y S - a_z\sigma_z) \, .
\end{equation}
Hence, $\Walk(A)\in\mathcal{T}_0$ requires $a_y=0$ and $w\,a_z +(1-w) \,a_x=0$. Consequently, we have $\Walk(A)=(w\,a_x-(1-w)\,a_z)\sigma_z$ and therefore $\Norm{\Walk(A)}^2=\vert w\,a_x-(1-w)\,a_z\vert^2\leq(w^2+(1-w)^2)\cdot(\vert a_x\vert^2+\vert a_z\vert^2)<\vert a_x\vert^2+\vert a_z\vert^2=\Norm{A}^2$, which proves the assertion.
\end{proof}
According to \eqref{Eq:FirstOrderFormula} the ballistic order vanishes, i.e. $\mu'=0$. The diffusive order can then be derived from equation \eqref{Eq:SecondOrderFinal} and corollary \ref{cor:CovMatPowerSeries}. The result is depicted in figure \ref{fig:BrokenLinkDiffusionConst} and shows that for the completely localized walk, i.e. $w=0$, the position probability becomes a point measure at the origin whereas for the usual Hadamard walk the diffusion constant $D$ diverges. The red points in the figure correspond to an approximation to the diffusion constant by calculating the variance of the position probability distributions after 100 time steps. For $w\leq 0.6$ the variance data points are in good agreement with the diffusion constant, for larger $w$, however, there is a clearly visible deviation of finite time variance and diffusion constant. This is the fingerprint of a generic behavior: when approaching a ballistic quantum walk, here $w\rightarrow 1$, the time scale, at which the crossover from ballistic to diffusive behavior happens, diverges. Hence, approximations with a fixed number of time steps become worse and worse in the limit $w\rightarrow 1$.

The plot in figure \ref{fig:BrokenLinkDiffusionConst} also shows a comparison of the diffusion constant and the guess $w/(1-w)$ by Romanelli \textit{et al.} for the model considered in \cite{Romanelli2}. The diffusion constant for our example shows a similar functional dependence on $w$, though there are small deviations.

\begin{figure}[htb]
\begin{center}
	\includegraphics[width=8cm]{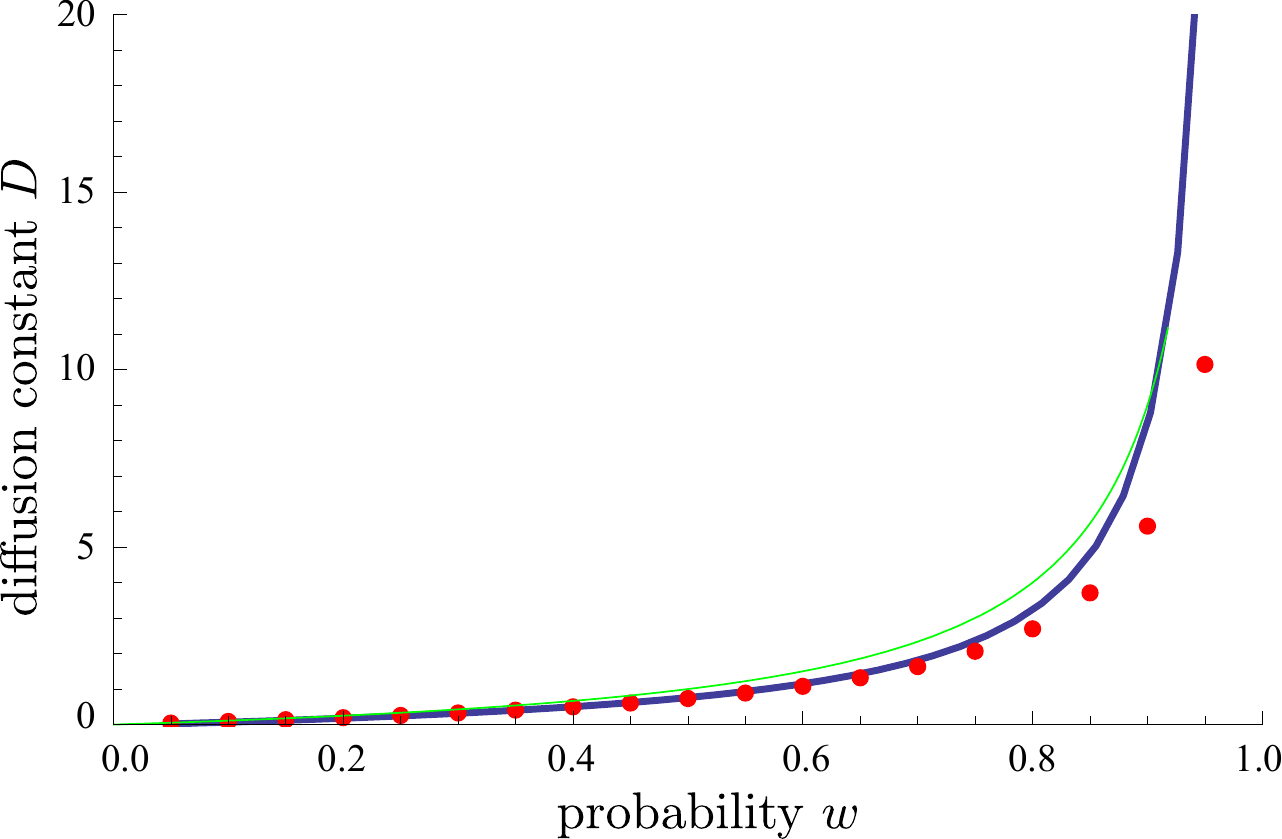}
\end{center}
\caption{The diffusion constant of the one-dimensional quantum walk with random reflections as a function of the probability $w$. The blue curve shows an approximation of the diffusion constant evaluated according to corollary \ref{cor:CovMatPowerSeries} for different values of $w$. The red dots show the variances of the position distribution for $t=100$ and the green line shows the guess $D=w/(1-w)$ by Romanelli \textit{et al.} for a similar model considered in \cite{Romanelli2}.}
\label{fig:BrokenLinkDiffusionConst}
\end{figure}

\subsection{Hadamard Walk with Dephasing}
The Hadamard walk with dephasing is given by a modification of the local coin operator $H$ in $W_H=S\cdot \idty\otimes H$, with $S$ and $H$  according to \eqref{Eq:Hadamard}. We modify the coin $H$ by an additional random relative phase shift between the states $\Ket{+}$ and $\Ket{-}$
\[
H_\varphi  =
\frac{1}{\sqrt{2}}
\left(
\begin{array}{cc}
\exp^{\ii \varphi} & \exp^{\ii \varphi}\\
\exp^{-\ii \varphi} & -\exp^{-\ii \varphi}
\end{array}
\right)\, .
\]
Note that, by introducing the Pauli operators $\sigma_x,\sigma_y$ and $\sigma_z$, we can write $H_\varphi=\exp^{\ii\varphi \sigma_z}H$. This phase shift is assumed to happen independently and identically in time and space, in other words, the phase $\varphi$ is chosen for each lattice site and in each time step according to a fixed probability measure $\nu$ on $[-\pi,\pi)$. In Fourier space we have the following expression for the Hadamard walk with dephasing:
\begin{equation}
\label{Eq:DephasingHadamard}
 \Walk (A)(p)=S^*(p)\Coin (A) S(p)\quad,\quad \Coin(A)=\int\nu(d\varphi)H\exp^{-\ii \varphi \sigma_z}\cdot A_0\cdot \exp^{\ii \varphi \sigma_z}H +\widetilde H^*(A-A_0) \widetilde H
\end{equation}
Here, $\widetilde H$ denotes the operator $\widetilde H=\int \nu(d\varphi) \,\exp^{\ii \varphi \sigma_z} H$. In the following we set $r_n=\vert\int\nu(d\varphi)\exp^{\ii\varphi n}\vert$ and $\theta_n=\arg\int\nu(d\varphi)\exp^{\ii\varphi n}$ and note that $\Coin$ only depends on $r_n$ and $\theta_n$ for $n=1,2$ and not on the actual form of the distribution $\nu(d\varphi)$, see \eqref{Eq:DephCoin1} and \eqref{Eq:DephCoin2}.

It is easy to see that
\[
H_{\varphi'}^*H_{\varphi}=H\exp^{\ii (\varphi-\varphi') \sigma_z}H=\exp^{\ii (\varphi-\varphi') \sigma_x}=\left(
\begin{array}{cc}
\cos (\varphi-\varphi') & \ii \sin (\varphi-\varphi')\\
\ii \sin (\varphi-\varphi') & \cos (\varphi-\varphi')
\end{array}
\right)\, ,
\]
which shows that the set of matrices $H_{\varphi'}^*H_\varphi$ is reducible. Nonetheless, we can apply theorem \ref{thm:MainThm} after we prove that $\Walk^2$ is strictly contractive on $\{\idty\}^\perp$.
\begin{lem}
Let $\Walk$ be according to \eqref{Eq:DephasingHadamard} and $\nu$ such that $r_1,r_2<1$. Then $\Walk^2$ is strictly contractive on $\{\idty\}^\perp$.
\end{lem}
\begin{proof}
We write the coin operator as $\Coin(A) =T(A_0) +\widetilde H^* (A-A_0)\widetilde H$ and express the dephasing in the maps $T$ and conjugation by $\widetilde H$ in the basis $\{\idty,\sigma_x,\sigma_y,\sigma_z\}$ of Pauli operators. The dephasing in $T$ corresponds to the matrix representation
\begin{equation}
\label{Eq:DephCoin1}
A\mapsto \int\nu(d\varphi)\exp^{-\ii \varphi \sigma_z}\cdot A\cdot \exp^{\ii \varphi \sigma_z} \quad\longrightarrow\quad
\left(
\begin{array}{cccc}
1 & 0 & 0 & 0\\
0 & r_2 \cos \theta_2 & -r_2 \sin \theta_2 &0\\
0 & r_2 \sin \theta_2 & r_2 \cos \theta_2 &0\\
0 & 0 & 0 & 1
\end{array}
\right)
\end{equation}
and the dephasing in the conjugation by $\widetilde H$ has matrix representation
\begin{equation}
\label{Eq:DephCoin2}
A\mapsto \int\nu(d\varphi)\exp^{-\ii \varphi \sigma_z}\cdot A\cdot \int\nu(d\varphi')\exp^{\ii \varphi' \sigma_z} \quad\longrightarrow\quad
r_1^2\cdot\left(
\begin{array}{cccc}
1 & 0 & 0 & 0\\
0 &  \cos 2\theta_1 & - \sin 2\theta_1 &0\\
0 &  \sin 2\theta_1 &  \cos 2\theta_1 &0\\
0 & 0 & 0 & 1
\end{array}
\right)\, .
\end{equation}
Clearly, if $r_1<1$ the conjugation by $\widetilde H$ is strictly contractive. If $r_2<1$ we have that $T$ is strictly contractive on the subspace spanned by $\sigma_x$ and $\sigma_y$. The operator $\sigma_z$ is mapped to $\sigma_x$ by $T$, and hence, $S^*T(\sigma_z)S\in\mathcal{T}_0^\perp$, which proves $\Norm{\Walk^2(A)}<\Norm{A}$ for all $A\perp\idty$.
\end{proof}
Since $\ind\,S=0$, the ballistic scaling yields a point measure at the origin as asymptotic position distribution. For the asymptotic position distribution in diffusive scaling we need to solve the equation
\[
 \Walk (A')(p)-A'(p)=
 \left(
 \begin{array}{cc}
 -\ii & 0 \\
 0 & \ii
 \end{array}
 \right)\, ,
\]
which we do by numeric approximation according to corollary \ref{cor:CovMatPowerSeries}.
\begin{figure}
\includegraphics[width=16cm]{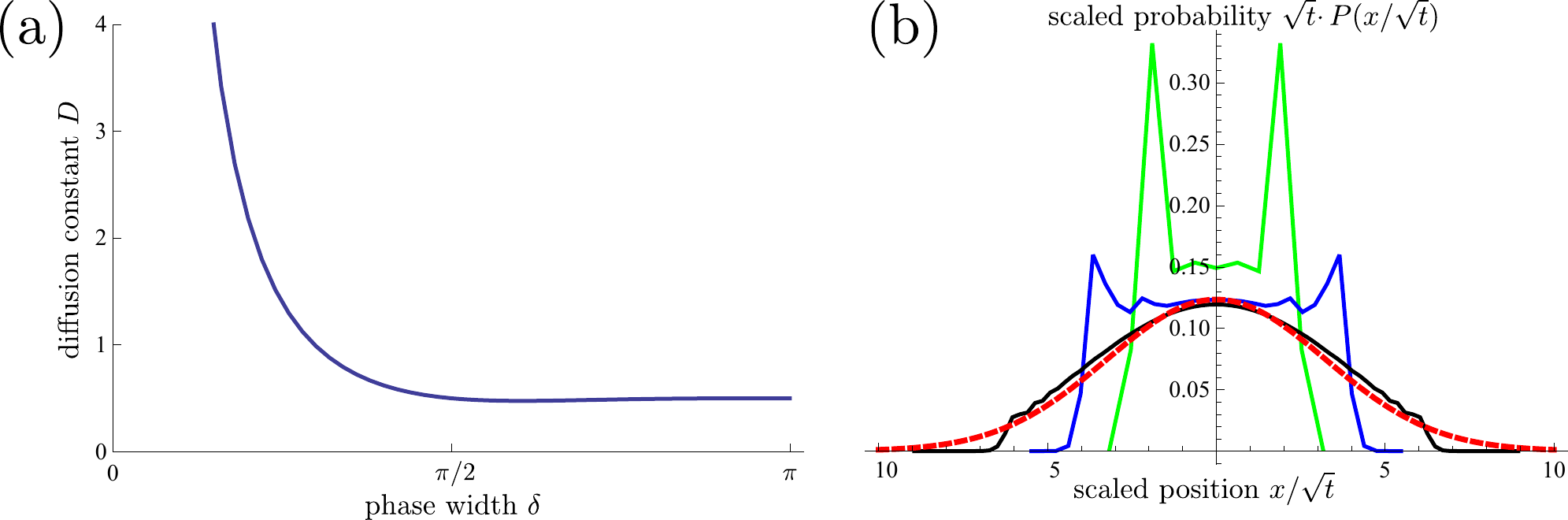}
\caption{(a) shows a plot of the diffusion constant $D$ depending on the width $\delta$ of the interval $I_\delta$ from which the phase $\varphi$ is chosen with uniform probability. For the plot shown in (b) we fixed $\delta=\pi/8$ and computed the position distribution after 10 (green), 30 (blue), and 80 (black) time steps with initial state $\rho_0=2^{-1}\idty$. For better comparison with the asymptotic position distribution (red-dashed line) the average over two neighboring lattice sites was computed such that the finite time step position distributions are non-zero everywhere.}
\label{Fig:DiffDephase}
\end{figure}
For concreteness let us consider a family of measures $\nu_\delta$ on $[-\pi,\pi)$, indexed by $\delta\in [0,\pi)$, and defined via
\[
\nu_\delta( d\varphi)=\frac{1}{2\delta}\chi_{I_\delta}(\varphi) d\varphi\,.
\]
Here, $\chi_{I_\delta}$ denotes the characteristic function of the set $I_\delta=[-\delta,\delta)$ and $d\varphi$ the Lebesgue measure. The measures $\nu_\delta$ represent uniform distributions of $\varphi$ on $I_\delta$ and figure \ref{Fig:DiffDephase} shows a plot of the diffusion constant and a comparison of the position distribution for a finite number of time steps with the asymptotic position distribution.

\subsection{Quantum Walk with Continuous Coin Distribution}

One of the goals of this paper is to take into account fluctuations of the coin operators in space and time as they occur in experiments. However, in most examples up to now, discrete distributions of the coin operators have been studied. For a model of realistic noise sources in experiments it seems to be more appropriate to assume a continuous distribution around some specific coin operation. Let us therefore consider a quantum walk on $\ell^2(\Integers )\otimes \Complex^2$ with the shift according to \eqref{Eq:Hadamard} and coin operator constructed from the following $2\pi$-periodic family of unitary matrices
\begin{equation}
	U_r=\begin{pmatrix}\cos r&\sin r\\\sin r&-\cos r\end{pmatrix} \, ,
\end{equation}
for which $U_{\pi/4}=H$, is the Hadamard coin \eqref{Eq:Hadamard}, whereas $U_{\pi/2}=\sigma_x$ and $U_0=\sigma_z$.
Taking on the idea of random fluctuations around some target coin operator, we consider
a gaussian probability distribution on $I(r_0)= [ r_0-\pi,r_0+\pi]$ around some point $r_0\in\Reals$.
The measure on the parametrization space is then given by
\begin{equation}
	\nu(X)=\frac{1}{N}\int\limits_{X} e^{-\frac{(r-r_0)^2}{\sigma^2}}\:dr ,\qquad
\label{eq:GaussianMeasure}
\end{equation}
where $X$ is a Borel set on $I(r_0)$, $dr$ is the usual Lebesgue measure and the normalization $N$ is chosen such that the measure is normalized on $I(r_0)$.
Now we can write the coin operator as
\begin{equation}
\label{Eq:GaussCoin}
	\Coin(A)	=\sum_{x,y}\Ketbra{x}{y}\otimes\left(\delta_{x,y}\int\limits_{I(r_0)} \frac{1}{N}dr\:e^{-\frac{(r-r_0)^2}{\sigma^2}}U_r^*A_{xy}U_r+(1-\delta_{x,y})\widetilde U^*A_{xy}\widetilde U\right),
\end{equation}
where
\begin{equation*}
	\widetilde U=\int\limits_{I(r_0)}\nu(dr)U_r=\int\limits_{I(r_0)} \frac{1}{N}dr\:e^{-\frac{(r-r_0)^2}{\sigma^2}}U_r.
\end{equation*}
A straightforward calculation shows $[U_r,U_{r'}]=-2i\sin(r-r')\sigma_y$, hence $\{U_r\}$ is irreducible whereas the set $\{U_{r'}^*U_r\}$ is reducible. Nonetheless, $\Walk^2$ is strictly contractive on $\{\idty\}^\perp$, as the following lemma shows.
\begin{lem}
Let $\Walk=S^*\Coin S$ be a quantum walk with $\Coin$ according to \eqref{Eq:GaussCoin} and $S$ according to \eqref{Eq:Hadamard}, then $\Walk^2$ is strictly contractive on $\{\idty\}^\perp$.
\end{lem}
\begin{proof}
Since the set $\{U_r\}$ is irreducible and all $U_r$ are hermitian we can exclude that $\widetilde U$ has eigenvalues of modulus one and it follows that $\Norm{\widetilde U}_{op}<1$. This proves the contractivity of the map corresponding to conjugation by $\widetilde U$.

Again, using Pauli operators as a basis for the set of two-dimensional matrices it is easy to see that
\[
U_r^*\sigma_x U_r = -\cos 2r \,\sigma_x + \sin 2r\, \sigma_z\quad ,\quad  U_r^*\sigma_y U_r=-\sigma_y\quad , \quad U_r^* \sigma_z U_r = \sin 2r \,\sigma_x + \cos 2r\, \sigma_z \, .
\]
Taking the expectation value of these equations with respect to the probability distribution $\nu$ shows that the diagonal part of $\Coin$, which we denote again by $T$, is strictly contractive on $\sigma_x$ and $\sigma_z$. The assertion follows from the fact that the conjugation with the shift operator $S$ maps $\sigma_y$ to $\mathcal{T}_0^\perp$.
\end{proof}

\begin{figure}[htb]
\begin{center}
\includegraphics[width=15cm]{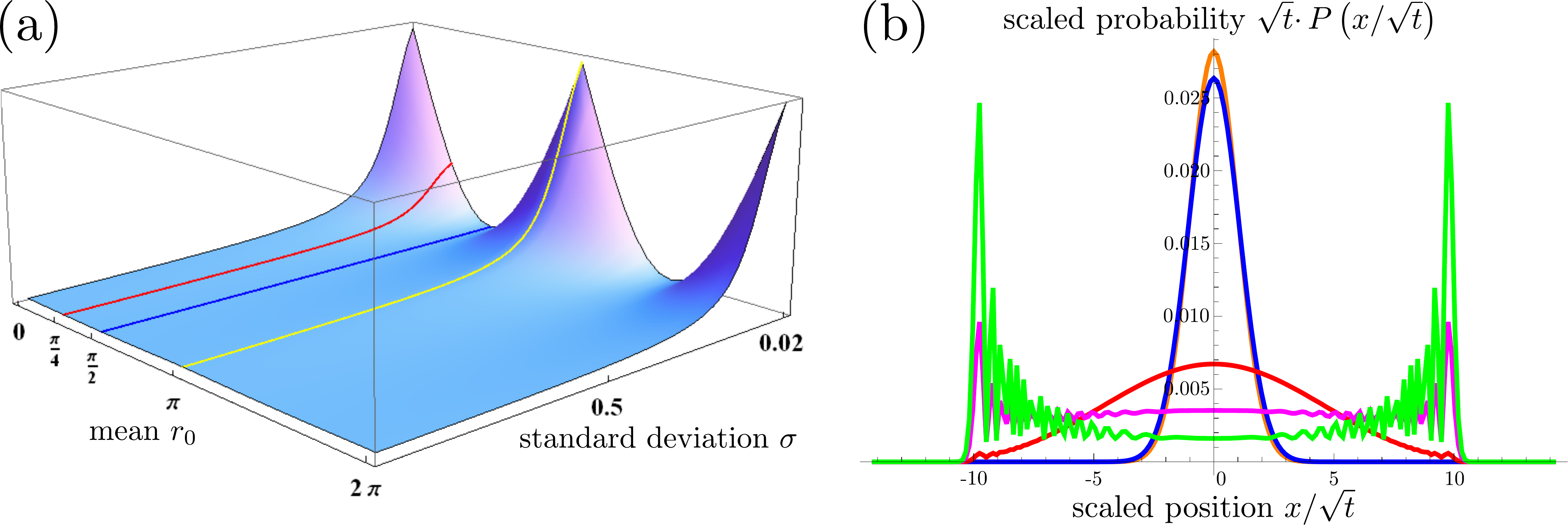}
\end{center}
\caption{Plot (a) shows the diffusion constant $D$ as a function of $\sigma\in(0,1]$ and $r_0\in[0,2\pi]$. For every fixed $\sigma$ the diffusion constant has a peak at $r_0\in\{0,\pi\}$ (yellow line) which correspond to measures peaked around $\pm\sigma_z$ coins. The minima occur at $r_0=\pi/2,3\pi/2$ (blue line) where the random coins are peaked around $\sigma_x$. The diffusion constant of a quantum walk peaked at the Hadamard coin can be found at $r_0=\pi/4$ (red line) and dephased variants of it at $r_0=k\pi/4$, $k\in\mathbbm Z$. Plot (b) shows the position probability distribution of the quantum walk with gaussian distribution peaked at the Hadamard coin $H$ for a number of $t=200$ time steps in diffusive scaling. Again, the average over two neighboring lattice sites was computed in order to remove the oscillating behavior of the probability distributions. The different plots correspond to $\sigma=2$ (orange), $\sigma=1$ (blue), $\sigma=0.2$ (red), $\sigma=0.1$ (magenta), and $\sigma=0.01$ (green). What can be seen is that for $\sigma\rightarrow0$, where the gaussian measure converges weakly to the point measure at the Hadamard coin, the diffusion constant diverges, which explains the Hadamard like position probability for $\sigma=0.01$. Nevertheless for $t\rightarrow\infty$ it becomes gaussian, as long as $\sigma>0$.
}
\label{fig:PosProbGaussianDiffConstGaussian}
\end{figure}

The shift operator satisfies $\ind \,S=0$, such that according to \eqref{Eq:FirstOrderFormula} $\mu'=0$ which means that the ballistic order vanishes. Hence, the characteristic function in diffusive scaling is well-defined and yields a Gaussian. The position probability distribution can be calculated by the inverse Fourier transform resulting again in a Gaussian. For a finite number of $t=200$ time steps the position probability distribution for a coin operator peaked at the Hadamard coin $U_{\pi/4}=H$ is depicted in figure \ref{fig:PosProbGaussianDiffConstGaussian} (b) for $\sigma\in\{0.01,0.1,0.2,1,2\}$. It is apparent that for small $\sigma$ the distribution looks like the one of the usual Hadamard coin plus some gaussian background. This is again an indication that for decreasing $\sigma$, i.e. coin distributions with decreasing width, the crossover from ballistic to diffusive behavior happens at later times. Indeed, since the gaussian measure \eqref{eq:GaussianMeasure} converges weakly to a point measure, i.e. $\nu\overset{w}{\longrightarrow}\delta_{r_0}$ for $\sigma\rightarrow0$, we expect a divergent diffusion constant as the translation invariant and unitary quantum walk with coin $U_{r_0}$ exhibits ballistic behavior.

Figure \ref{fig:PosProbGaussianDiffConstGaussian} (a) depicts the diffusion constant $D$ as a function of $\sigma\in(0,1]$ and $r_0\in[0,2\pi]$. Apparently, $D<\infty$ for $\sigma>0$ independent of the coin at which the distribution is peaked. The symmetry around $r_0=\pi$ comes from the fact that we have $U_{\pi+r}=\sigma_z U_{\pi-r}\sigma_z$. By $\Sigma_z$ we denote the operator on $\WSp$ corresponding to the Fourier transform $\sigma_z$. Since $\Sigma_z$ commutes with the shift $S$ we have the relation $S^*\Sigma_z\Coin (\Sigma_z\exp^{\ii \varepsilon\lambda\cdot Q}\Sigma_z)\Sigma_z S =\Sigma_zS^*\Coin (\exp^{\ii \varepsilon\lambda\cdot Q}) S \Sigma_z$. Denoting the walk operator with additional $\Sigma_z$ rotations by $\widetilde \Walk $, i.e. $\widetilde \Walk (A) =S^*\Sigma_z\Coin(\Sigma_z A\Sigma_z)\Sigma_z S$ we get the relation
\[
\widetilde\Walk^t(\exp^{\ii \varepsilon\lambda\cdot Q}) = \Sigma_z\Walk^t(\exp^{\ii\varepsilon\lambda\cdot Q})\Sigma_z \, .
\]
Changing the parameter from $r_0=\pi+r$ to $r_0=\pi-r$ is equivalent to changing the initial state from $\rho_0$ to $\Sigma_z\rho_0\Sigma_z$. Since the characteristic function in the asymptotic limit does not depend on the initial state (\eqref{Eq:CharFuncBall} and \eqref{Eq:CharFuncDiff}), the diffusion constant is independent of the above change of $r_0$. The periodic dependence on $r_0$ reflects the periodicity of the $U_r$. The minima on the slices with $\sigma$ fixed are the ones where the peak of the gaussian measure \eqref{eq:GaussianMeasure} is at $U_{r_0}=\sigma_x$, the reason being that the unitary and translation invariant quantum walk with coin $\sigma_x$ shows no propagation at all. The maxima correspond to $r_0\in\{0,\pi\}$ where the measure is peaked around $\sigma_z$.

\subsection{Quantum Walks in Two Dimensions}
For a quantum walk on a two-dimensional lattice, the Hilbert space is $\ell^2(\Integers^2)\otimes\KK $. That is, the position is given by two-component vectors $x=(x_1,x_2)$. We consider a four-dimensional coin space $\KK\cong \Complex^4\cong \Complex^2\otimes \Complex^2$. The shift operator conditioned on the internal degree of freedom is given by

\begin{equation}\label{Eq:2Dshiftop}
S(p)=\begin{pmatrix}e^{iv_{\uparrow}\cdot p}&0&\dots&\\0&e^{iv_{\downarrow}\cdot p}&&\\\vdots&&e^{iv_{\leftarrow}\cdot p}&\\&&&e^{iv_{\rightarrow}\cdot p}\end{pmatrix},\quad p=(p_1,p_2)\, ,
\end{equation}
where $v_{\uparrow}=(1,0)$, $v_{\downarrow}=(-1,0)$, $v_{\leftarrow}=(0,-1)$ and $v_{\rightarrow}=(0,1)$. The explicit example we are going to consider is a quantum walk with coin operation constructed from the set of unitaries $\{ U_1=H\otimes H,\: U_2=\sigma_x\otimes\idty \}$, where $U_1$ is applied with probability $(1-w)$ and $U_2$ occurs with probability $w$. In order to apply theorem \ref{thm:MainThm}, some power of the walk operator $\Walk$ must be strictly contractive on the subspace $\{\idty\otimes\idty\}^\perp$. Apparently, $\Walk$ itself is not strictly contractive on $\{\idty\otimes\idty\}^\perp$ because the $p$-independent operator $\idty\otimes H$ is an eigenvector of $\Coin$ with eigenvalue $1$. Moreover, the coins $U_1$ and $U_2$ are reducible. Denoting the eigenvectors of $H$ by $\Ket{\pm}$ it is easy to see that the two dimensional subspaces $\Complex^2\otimes\Ket{\pm}$ are invariant subspaces for $U_1$ and $U_2$. To confirm the applicability of theorem \ref{thm:MainThm} we now prove that $\Walk^2$ is strictly contractive.
\begin{lem}
Let $w\in (0,1)$ and $\Walk=S^*\Coin S$, with $S$ according to \eqref{Eq:2Dshiftop} and $\Coin$ be defined by the coins $U_1=H\otimes H$ and $U_2=\sigma_x\otimes H$ of which $U_1$ is applied with probability $1-w$ and $U_2$ is applied with probability $w$. Then $\Walk^2$ is strictly contractive on $\{\idty\otimes\idty \}^\perp$.
\end{lem}
\begin{proof}
Although the coins $U_1$ and $U_2$ are reducible we can exclude the existence of a common eigenvector by calculating the determinant of their commutator. If there is a common eigenvector this determinant equals zero, but since we have $\det [U_1,U_2]=-1\cdot \det[H,\sigma_x]=-2$ such an eigenvector cannot exist. This already proves that the operator norm of the hermitian matrix $\widetilde U= (1-w)\,U_1+w \,U_2$ is strictly less than $1$, hence, conjugation by $\widetilde U$ is a strictly contractive map.

Again, let $T$ denote the diagonal part of the coin operator $\Coin$, i.e.,
\[
T(A_0)=(1-w)\,H\otimes H \cdot A_0\cdot H\otimes H + w\,\sigma_x\otimes\idty\cdot A_0\cdot \sigma_x\otimes\idty\, .
\]
By choosing an orthonormal operator basis for $\mathcal{T}_0$ and calculating the action of $T$ with respect to this basis one confirms that $T$ is a hermitian map. Clearly, $\Norm{T}_{op}=1$ and in order to verify the contractivity of $\Walk^2$ we have to consider the eigenvectors of $T$ which are orthogonal to $\idty\otimes\idty$ and have eigenvalues $\pm 1$. It is sufficient to prove that conjugation by $S$ maps these eigenvectors to vectors with non-zero overlap with the subspace $\mathcal{T}_0^\perp$. These eigenvectors have to be common eigenvector of the maps corresponding to conjugation with $U_1$ and $U_2$, hence, we find exactly $\idty\otimes H$, $\sigma_y\otimes \idty$ and $\sigma_y\otimes H$ as eigenvectors with eigenvalues $\pm 1$. It is easy to see that the only vectors in $\mathcal{T}_0$ which are mapped to $\mathcal{T}_0$ again are all $4$-dimensional diagonal matrices, which proves the assertion.
\end{proof}
The argument of the characteristic function is a two-dimensional vector $\lambda=(\lambda_1,\lambda_2)^T$, hence,
\begin{equation}
\Lambda=\begin{pmatrix}\lambda_1&0&\dots&\\0&-\lambda_1&&\\\vdots&&-\lambda_2&\\&&&\lambda_2\end{pmatrix}\,.
\end{equation}
Since $\sum_i v_i=0$, the ballistic order is zero by \eqref{Eq:FirstOrderFormula}, the second order correction $\mu''$ can be written as quadratic form in $\lambda$
\begin{equation}
\mu''=-\lambda^T\cdot D\cdot\lambda,
\end{equation}
with the covariance matrix $D$. By theorem \ref{thm:MainThm} we can compute the matrix elements of $D$ via \eqref{Eq:CovarianceMatrix} and \eqref{Eq:Ralphabeta}. We use again corollary \ref{cor:CovMatPowerSeries} to get a numerical approximation of $D$.

The asymptotic probability distribution can be computed via the characteristic function
\begin{equation}
C(\lambda)=\exp^{-\frac{1}{2}\lambda^T D\lambda}\,
\end{equation}
and the probability at position $x=(x_1,x_2)$ is given by its inverse Fourier transform
\begin{align}
P(x)&=\frac{1}{(2\pi)^2}\int_{\mathbbm R^2}\text{d}\lambda\:C(\lambda)e^{-i\lambda x}\\
&=\frac{1}{(2\pi)^2}\int_{\mathbbm R^2}\text{d}\lambda\:
e^{-\frac{1}{2}\lambda^T D\lambda-ix\lambda}\nonumber\\
&=\frac{1}{2\pi\sqrt{\det D}}\exp^{-\frac{1}{2}x^T D^{-1}x}\, , \nonumber
\end{align}
see e.g. \cite{negele1998quantum}, with
\begin{equation*}
D^{-1}=\frac{1}{\text{det} D}\begin{pmatrix}D_{22}&-D_{12}\\-D_{12}&D_{11}\end{pmatrix}.
\end{equation*}
The asymptotic probability distribution for our example is illustrated in figure \ref{Fig:2DAsymptotic} for two different values of $w$.
Here $U_1=H\otimes H$ renders a Hadamard-walk in two dimensions and $U_2=\sigma_x\otimes\idty$ yields propagation strictly along the diagonal. The coin $U_2$ in combination with $U_1$ leads to significantly increased spreading in the diagonal direction and reduced spreading in the anti-diagonal direction.

\begin{figure}
\centering
\includegraphics[width=12cm]{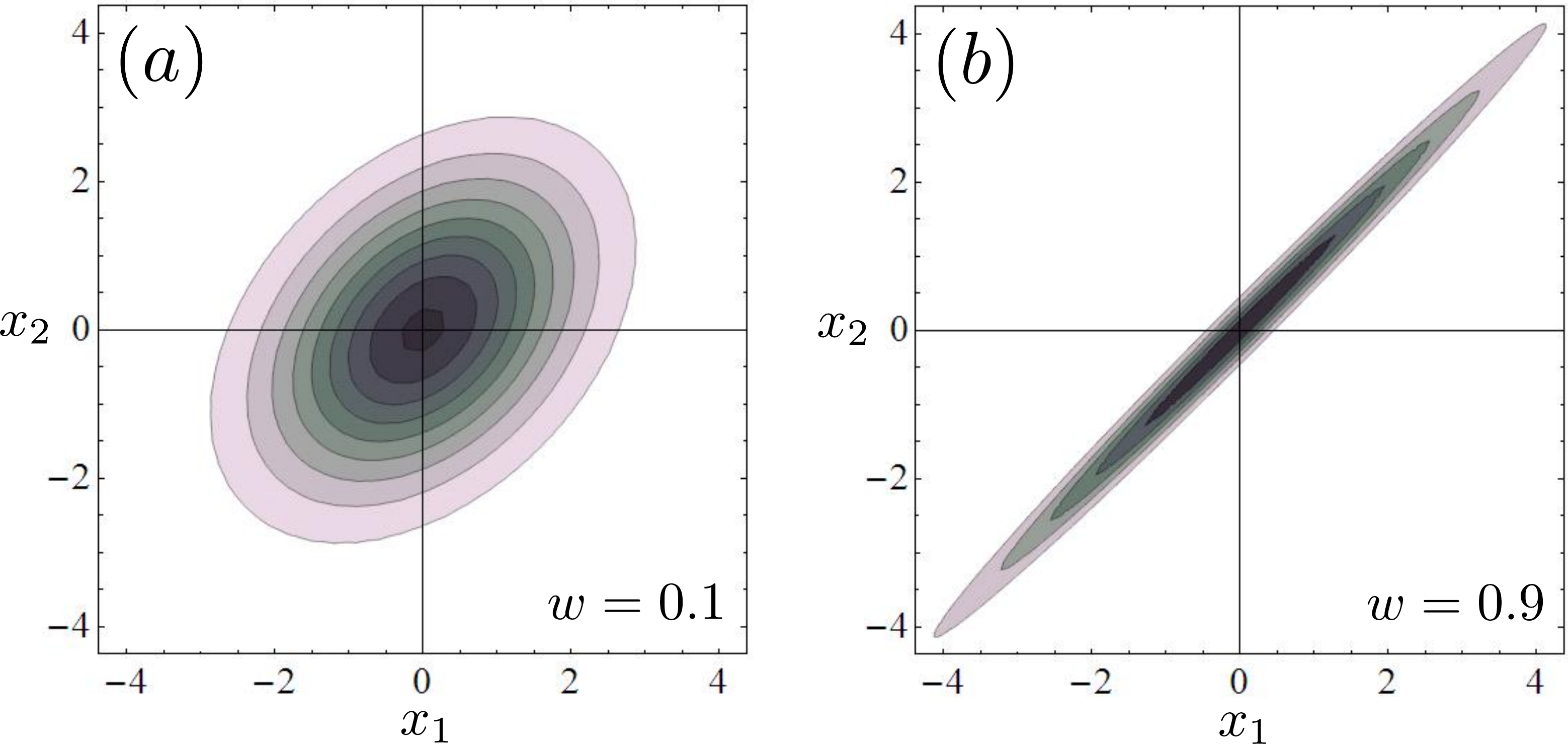}
\caption{Contour plots of the asymptotic position probability distributions of the quantum walk with coins $ U_1=H\otimes H$ and $U_2=\sigma_x\otimes\idty$ with probabilities $\{(1-w),w \}$ for (a): $w=0.1$ and (b): $w=0.9$. The coin $U_1$ alone would lead to a two-dimensional Hadamard walk, but $U_2$ incorporates propagation strictly along the diagonal.
}
 \label{Fig:2DAsymptotic}
\end{figure}

\section{Generalizations}\label{sec:generalizations}
Although the vast majority of literature considers quantum walks which are a composition of a single coin and a single shift operator we will briefly comment on more general models of quantum walks. A way to characterize quantum walks more abstractly is to define them to be discrete time evolutions on a lattice $\HH=\WSp$, which are local and translation invariant. This definition is clearly satisfied if several quantum walks $\Walk_i$ are concatenated and considered as a single time step $\Walk$ given by
\begin{equation}
\label{Eq:GenWalk}
\Walk=\Walk_n\circ \ldots \circ \Walk_1\,.
\end{equation}
More precisely, we assume that each $\Walk_i$ is a quantum walk according to \eqref{Eq:WalkPosSpace}, which means it can be written as
\[
\Walk_i (A)= S_i^* \Coin_i (A) S_i\, ,
\]
where $\Coin_i$ is a, possibly decoherent, coin operator and $S_i$ is a unitary state dependent shift operator. For this generalized model of quantum walks we have the following proposition, which also covers the extremal case where several unitary quantum walks are concatenated with one decoherent quantum walk.
\begin{prop}
Let $\Walk$ be a generalized quantum walk according to \eqref{Eq:GenWalk}. If at least one $\Walk_i$ is strictly contractive on $\{\idty\}^\perp$ we can apply perturbation theory to the eigenvector equation $\Walk_\varepsilon (A_\varepsilon)=\mu_\varepsilon A_\varepsilon$ to determine the asymptotic position distribution of $\Walk$. In particular, let $\ind\,S_i$ denote the index of $S_i$ according to \eqref{Eq:Index}, the asymptotic position distribution of $\Walk$ in ballistic scaling is given by a point measure at $(\dim\KK)^{-1}\sum_{i=1}^n  \ind\,S_i $.
\end{prop}
\begin{proof}
Clearly, $\Walk$ satisfies $\Walk(\idty)=\idty$ and since $\Walk$ maps $\idty$ and $\{\idty\}^{\perp}$ to orthogonal subspaces it follows that $\Walk$ is strictly contractive on $\{\idty\}^{\perp}$ if at least one of the $\Walk_i$ is strictly contractive on $\{\idty\}^{\perp}$. Hence, the non-degeneracy of the eigenvalue $1$ of $\Walk$ is assured.

The proof of the analyticity of $\Walk_\varepsilon$ is similar to the case in theorem \ref{thm:MainThm}, hence, $\Walk_\varepsilon$ satisfies the requirements of the Kato-Rellich theorem \ref{thm:KatoRellich}. This implies that the asymptotic behavior of $\Walk$ can be determined using our perturbation method.

The first order of the power series expansion of $\Walk_\varepsilon (A_\varepsilon)=\mu_\varepsilon A_\varepsilon$ reads
\[
\Walk_n'\circ \ldots \circ\Walk_1(\idty) + \ldots + \Walk_n\circ \ldots \circ \Walk_1'(\idty)  = \mu'\idty + A' - \Walk (A')
\]
and since $\Walk_i(\idty)=\idty$ and $\Walk_i'(A)(p)=\Walk_i (A)(p)\cdot \ii\Lambda_i$ this simplifies to
\[
\ii \Lambda_n+ \ii\sum_{i=1}^{n-2} \Walk_n \circ \ldots \circ \Walk_{n-i} (\Lambda_{n-i-1})(p) = \mu'\idty + A'(p) - \Walk (A')(p)\, .
\]
The scalar product of this equation with the unperturbed eigenvector $\idty$ yields again
\[
\mu'=\frac{\ii}{\dim\KK}\sum_{i=1}^n \tr (\Lambda_i)=\lambda\cdot\frac{\ii}{\dim \KK}\sum_{i=1}^n   \ind\,S_i \, .
\]
\end{proof}
The diffusive scaling of $\Walk$ can also be determined by equating coefficients of the perturbation expansion. Of course, the equations get more involved, but the general structure of the problem is the same. For example, the second order correction $\mu''$ to the eigenvalue $\mu_\varepsilon$ is again a quadratic form in $\lambda$ with constant coefficients and the asymptotic distribution in diffusive scaling is just a Gaussian independent of the initial state $\rho_0$.

To illustrate this, consider a concatenation two quantum walks $\Walk=\Walk_2\circ \Walk_1$. The second order equation reads
\[
\Walk (A'') + 2\Walk_2' \circ \Walk_1 (A')+2\Walk_2 \circ \Walk_1' (A') +\Walk_2'' (\idty) + \Walk_2\circ\Walk_1''(\idty) + 2 \Walk_2'\circ\Walk_1' (\idty)= \mu''\idty +2 \mu'A' + A''
\]
and by assuming $\mu'=0$ and taking the scalar product with $\idty$ again, we get
\[
\mu'' = \frac{1}{(2\pi)^s \dim\KK}\int dp^s \tr ( -\Walk_2(\Lambda_1^2)(p)-\Lambda_2^2 -2 \Walk_2(\Lambda_1)(p)\cdot\Lambda_2 + 2 \ii (\Walk_2(\Walk_1(A')(p)\cdot\Lambda_1)(p)+\Walk_2(\Walk_1(A'))(p)\cdot\Lambda_2) ) \, ,
\]
which is a quadratic form in $\lambda$ with no further parameter dependencies.

\section{Conclusion}\label{sec:conclude}
We have shown that quantum walks with spatio-temporal fluctuations of the local coin operator exhibit, under rather mild assumptions, diffusive behavior in the long-time limit. Our method provides complete information about the asymptotic position distribution of the considered quantum walk and, though the appearing equations may not be solvable in a simple manner, an approximation to the solution can always be found by computing a truncation of a power series expression for the exact solution.

The model of quantum walks with spatio-temporal coin fluctuations considered exhibits two generic features. First, the asymptotic position distribution in ballistic scaling is always given by a point measure. Secondly, for the asymptotic position distribution in diffusive scaling we get a Gaussian which is independent of the initial state.

One crucial assumption in our model is that the coins at different lattice sites and at different times are distributed identically and independently. Correlations of the coins in space or time have been studied in the literature and it was found that they can lead to different phenomena like ballistic or sub-ballistic behavior or even localization. However, a complete theory, combining correlations in time and space and providing sufficiently general criteria for different kinds of asymptotic behavior, is still missing.

The second assumption is of a more technical nature and concerns the irreducibility of the coin operators. This is similar to the case considered in \cite{timerandom}, where it was also shown that the extremest form of reducibility, namely commuting coin operators, leads again to ballistic behavior. Our method can, in principle, also be applied to quantum walks with spatio-temporal coin fluctuations and reducible coins by using degenerate perturbation theory.

\section{Acknowledgments}
We gratefully acknowledge support by the DFG (Forschergruppe 635) and the EU (CoQuit).

\bibliographystyle{alpha}
\bibliography{walkslit}

\end{document}

%% file: preamble.tex
\newcommand{\ii}{\ensuremath{\mathrm{i}}}
\renewcommand{\exp}{\ensuremath{e}}

\newcommand{\Reals}{\ensuremath{\mathbb{R}}}
\newcommand{\Integers}{\ensuremath{\mathbb{Z}}}
\newcommand{\Complex}{\ensuremath{\mathbb{C}}}
\newcommand{\Naturals}{\ensuremath{\mathbb{N}}}

\newcommand{\Walk}{\ensuremath{{\bf W}}}
\newcommand{\Coin}{\ensuremath{{\bf C}}}

\newcommand{\idty}{\ensuremath{\mathbbm{1}}}

\newcommand{\ind}{\ensuremath{\mathrm{ind}}}

\newcommand{\Ket}[1]{\ensuremath{\vert #1\rangle}}

\newcommand{\Ketbra}[2]{\ensuremath{\vert #1\rangle\langle#2\vert}}
\newcommand{\Scp}[2]{\ensuremath{\langle #1\vert #2\rangle}}
\newcommand{\ScpWOp}[3]{\ensuremath{\langle #1\vert #2\vert #3\rangle}}
\newcommand{\Norm}[1]{\ensuremath{\Vert #1\Vert}}
\newcommand{\tr}{\ensuremath{\mathrm{tr}}}

\newcommand{\WSp}{\ensuremath{\ell^2(\Integers^s)\otimes \KK}}
\newcommand{\HH}{\ensuremath{\mathcal{H}}}
\newcommand{\KK}{\ensuremath{\mathcal{K}}}

\newcommand{\BB}{\ensuremath{\mathcal{B}}}
\newcommand{\UU}{\ensuremath{\mathcal{U}}}

\newcommand{\TOps}{\ensuremath{\mathcal{T}_{\KK,s}}}

\newcommand{\Fourier}{\ensuremath{\mathcal{F}}}
